\documentclass[a4paper, 11pt]{article}
\usepackage[utf8]{inputenc}
\usepackage[english]{babel}
\usepackage{amsmath,  amssymb, amsfonts}
\usepackage{amsthm}
\usepackage{mathrsfs}
\usepackage{color}
\usepackage{graphicx}
\usepackage{enumerate}
\usepackage{bm}
\usepackage{cite}
\usepackage{url}
\usepackage{float}
\usepackage{algorithm} 
\usepackage{algorithmic} 

\newcommand{\R}{\mathbb{R}}
\newcommand{\Z}{\mathbb{Z}}

\DeclareMathOperator{\cp}{cp}

\newcommand{\hb}{\mathbf{h}}

\newcommand{\xb}{\mathbf{x}}
\newcommand{\yb}{\mathbf{y}}
\newcommand{\zb}{\mathbf{z}}
\newcommand{\alphab}{\bm{\alpha}}
\newcommand{\betab}{\bm{\beta}}

\newcommand{\Hb}{\mathbf{H}}

\newcommand{\Ib}{\mathbf{I}}
\newcommand{\Mb}{\mathbf{M}}

\newcommand{\Fb}{\mathbf{F}}

\newcommand{\Pb}{\mathbf{P}}

\newcommand{\Ab}{\mathbf{A}}
\newcommand{\Bb}{\mathbf{B}}

\newcommand{\CC}{\mathrm{C}}

\newcommand*\mcap{\mathbin{\mathpalette\mcapinn\relax}}
\newcommand*\mcapinn[2]{\vcenter{\hbox{$\mathsurround=0pt
  \ifx\displaystyle#1\textstyle\else#1\fi\bigcap$}}}
\newcommand*\mcup{\mathbin{\mathpalette\mcupinn\relax}}
\newcommand*\mcupinn[2]{\vcenter{\hbox{$\mathsurround=0pt
  \ifx\displaystyle#1\textstyle\else#1\fi\bigcup$}}}

\newcommand\floor[1]{\lfloor#1\rfloor}
\newcommand\ceil[1]{\lceil#1\rceil}

\DeclareFontFamily{OT1}{pzc}{}

\DeclareFontShape{OT1}{pzc}{m}{it}{<-> s * [1.200] pzcmi7t}{}

\DeclareMathAlphabet{\mathpzc}{OT1}{pzc}{m}{it}

\newtheorem{proposition}{Proposition}
\newtheorem{theorem}{Theorem}
\newtheorem{lemma}{Lemma}

\newtheorem{corollary}{Corollary}
\newtheorem{definition}{Definition}
\newtheorem{assumption}{Assumption}

\newtheorem{remark}{Remark}

\title{\LARGE \bf Clique Gossiping}

\author{Yang Liu, Bo Li, Brian D. O. Anderson, and  Guodong Shi\thanks{Y. Liu, B. D. O. Anderson, and   G. Shi are with the Research School of Engineering, The Australian National University, ACT 0200,
Canberra, Australia. B. Li is with  Key Lab of Mathematics Mechanization, Chinese Academy of Sciences, Beijing 100190, China.   Email: yang.liu@anu.edu.au, libo@amss.ac.cn,  brian.anderson@anu.edu.au, guodong.shi@anu.edu.au.}
}
\date{}

\begin{document}

\maketitle

\begin{abstract}
This paper  proposes and investigates a framework for clique gossip protocols. As complete  subnetworks, the existence of cliques  is ubiquitous in various social, computer, and engineering networks. By clique gossiping, nodes interact with each other along a sequence of cliques.  Clique-gossip protocols are defined as arbitrary  linear node   interactions where node states are vectors evolving as linear dynamical  systems.  Such protocols become clique-gossip averaging algorithms when node states are scalars under averaging rules. We generalize the classical  notion of line graph  to capture the essential  node interaction  structure induced by both the underlying network and the specific clique sequence. We prove a fundamental eigenvalue invariance principle  for   periodic clique-gossip protocols, which  implies  that any permutation of the clique sequence leads to the same spectrum for the overall  state transition when the generalized line graph contains no cycle. We also prove that for a network with $n$ nodes, cliques with smaller sizes determined by factors of $n$ can always be constructed leading to finite-time convergent clique-gossip averaging  algorithms, provided $n$ is not a prime number. Particularly,  such finite-time convergence can be achieved with cliques of equal size $m$ if and only if $n$ is divisible by $m$ and they have exactly the same prime factors. A  proven fastest  finite-time convergent clique-gossip algorithm is constructed for clique-gossiping using size-$m$ cliques. Additionally, the  acceleration effects of clique-gossiping are illustrated via numerical examples.
\end{abstract}

\section{Introduction}

Gossip protocols provide a scalable and self-organized way of carrying out information dissemination over networks in the absence of centralized decision makers \cite{Eugster2003,Jelasity2005,Shah2008,ravelomanana2007optimal,hopkinson2009adaptive,iwanicki2010gossip}. In a gossip process, a  pair of nodes is selected randomly or deterministically at any given time, and then this pair of nodes {\em gossip} by exchanging information between each other as  a fundamental resource allocation protocol for computer networks \cite{Demers1987,Kempe2004}. Today, gossip processes are natural models for interpersonal interactions and opinion evolutions in social networks \cite{Doerr2012};  distributed systems running gossip protocols have been developed  to realize in-network control \cite{Bullo2012}, filtering \cite{Moura2010}, signal processing \cite{Rabbat2013}, and computation \cite{mou2017eigenvalue}.

Particularly, gossip averaging algorithms serve as a basic model for gossip protocols, where during one gossip interaction the two involved nodes average their states, which are simply real numbers \cite{Kempe2003,Boyd2006}. The  rate of convergence of  such gossip averaging algorithms can represent the performance of gossip protocols that are built on top of that, and quantify efficiency and robustness of  the underlying  network structure.
For random gossip algorithms, various results reveal that the network structure plays a critical role in shaping the convergence speed in the asymptotic sense \cite{Kempe2003,Boyd2006}. With deterministic gossiping, scheduling of the gossiping pairs becomes equally  influential  \cite{liu2011deterministic}; indeed, even finite-time convergence can be achieved with suitable number of nodes \cite{Guodong2016ton}. It is worth mentioning  that in certain cases transitions can be made  between deterministic and random gossip algorithms, where the Borel-Cantelli lemma provides immediate connections \cite{Shi-TIT}.

In this paper, we propose and investigate a framework involving clique gossip protocols, where simultaneous node interactions can take place among cliques instead of being restricted to pairs. Cliques are subnetworks that form a complete graph in their local topologies, whose existence is universal in social, computer, and engineering networks.  In fact, the use of cliques for beamforming and clustering has been employed in wireless sensor networks \cite{zeng2013clique,biswas-2013}. In a general model, each node holds a vector state at any given time, and clique-gossip protocols are arbitrary linear node dynamical interactions  along  a sequence of cliques that forms a coverage of the underlying network. When the node state vector is one-dimensional and the node interaction rules are  simply averaging, the clique-gossip protocol is reduced to a clique-gossip averaging algorithm. To facilitate the analysis of the network structure that governs the node interactions, we first generalize the classical notion of line graph in graph theory. Then our contributions are made through a few important convergence properties of  clique gossiping:
\begin{itemize}
\item We prove a fundamental eigenvalue invariance principle  for scheduling  periodic clique-gossip protocols,  valid for arbitrary clique-gossip protocols represented by linear dynamical systems. Such an invariance principle implies  that any permutation of the clique sequence leads to the same spectrum for the overall  state transition matrix if the generalized line graph associated with the clique-gossip algorithm contains no cycle.

\item We prove that for a network with $n$ nodes, there always exist ways of constructing cliques with smaller sizes leading to finite-time convergent clique-gossip averaging  algorithms, provided that $n$ is not a prime number. We also prove that such finite-time convergence can be achieved with cliques of equal size $m$ if and only if $n$ is divisible by $m$ and they have exactly the same prime factors.
\end{itemize}
It is worth mentioning that for clique gossiping with equal size $m$ cliques, we have constructed one of the fastest  finite-time convergent clique-gossip algorithms, which is shown to reach the fundamental complexity lower bound by elementary number theory. Additionally, we illustrate how multi-clique gossiping can be built over an existing clique-gossiping protocol, and  the acceleration effects of clique-gossiping are shown  via numerical examples.

The remainder of this paper is organized as follows. Section \ref{sec:model} presents the clique-gossip model. Section~\ref{sec:invariance} studies periodic clique-gossip protocols by establishing the eigenvalue invariance principle and investigating the rate of convergence. Section \ref{sec:finitetime} focuses on the possibilities of finite-time convergence for clique-gossip averaging algorithms. Finally Section \ref{sec:conclusions} concludes the paper with a few remarks on potential future directions.

\section{Problem Definition}\label{sec:model}
\subsection{Network Model}
Consider a group of nodes whose interaction structure is described by a simple undirected graph $\mathrm{G}=(\mathrm{V},\mathrm{E})$, where $\mathrm{V}=\{1,2,\dots,n\}$ is the node set and an element $(i,j)\in\mathrm{E}$ is an unordered pair of two distinct nodes $i,j\in\mathrm{V}$. Define the neighbor node set $\mathrm{N}_i$ of node $i$ by $\mathrm{N}_i=\{j:(i,j)\in\mathrm{E}\}$. Associated with a node subset $\mathrm{S}\subset\mathrm{V}$, its induced graph $\mathrm{G}[\mathrm{S}]$ is defined as the graph with node set $\mathrm{S}$ and the edge set containing all edges in $\mathrm{E}$ with both endpoints in $\mathrm{S}$. A clique $\mathrm{C}$ is a subset of $\mathrm{V}$ whose induced graph $\mathrm{G}[\CC]$ is a complete graph. Let $\mathsf{H}_{_{\mathrm{G}}}$ be the set containing all the cliques of $\mathrm{G}$. We index the elements in $\mathsf{H}_{_{\mathrm{G}}}$ as $\CC_1,\dots,\CC_D$. We say that two cliques $\CC_i,\CC_j\in\mathsf{H}_{_{\mathrm{G}}}$ with $i\neq j$ are adjacent if $\CC_i\mcap\CC_j\neq\emptyset$.

\begin{figure}[h]
\centering
\includegraphics[width=4.2in]{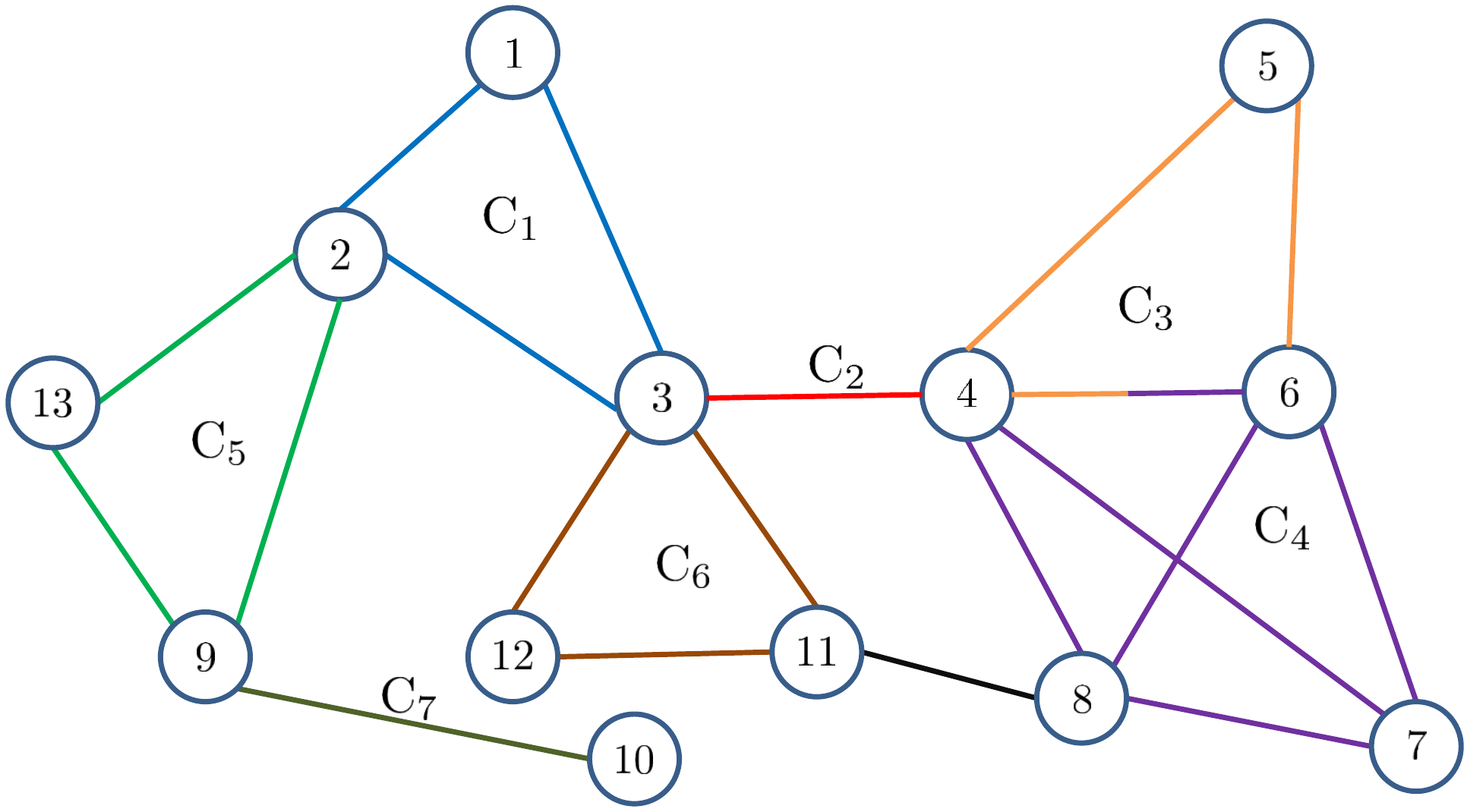}
\caption{A connected graph $\mathrm{G}$. Here $\CC_1=\{1,2,3\},\CC_2=\{3,4\},\CC_3=\{4,5,6,7\},\CC_4=\{4,6,7,8\},\CC_5=\{2,9,13\},\CC_6=\{3,11,12\},\CC_7=\{9,10\}$. Clearly $\{\CC_1,\CC_2,\dots,\CC_7\}$ is a clique coverage. Links within the same clique are marked with the same color and style.}
\label{fig:ori}
\end{figure}

\begin{definition}A subset $\mathsf{H}_{_{\mathrm{G}}}^\ast=\{\CC_{\mu_1},\dots,\CC_{\mu_d}\}\subset\mathsf{H}_{_{\mathrm{G}}}$ with $d\ge 1$ is called a clique coverage of $\mathrm{G}$ if $\bigcup_{\CC_i\in\mathsf{H}_{_{\mathrm{G}}}^\ast}\CC_i=\mathrm{V}$ and the union graph $\bigcup_{\CC_i\in\mathsf{H}_{_{\mathrm{G}}}^\ast}\mathrm{G}[\CC_i]$ is connected.
\end{definition}

Note that every connected graph has a clique coverage. Let $\mathsf{H}_{_{\mathrm{G}}}^\ast$ be the collection of pairs of the endpoints of all edges in $\mathrm{G}$. Then $\mathsf{H}_{_{\mathrm{G}}}^\ast$ is clearly a clique coverage.

\subsection{Clique-gossip Protocols}\label{sec:clique-gossip protocol}
Let $\mathsf{H}_{_{\mathrm{G}}}^\ast=\{\CC_{\mu_1},\dots,\CC_{\mu_d}\}$ be a clique coverage of $\mathrm{G}$ where $\mu_k\in\{1,\dots,D\}$ for $k=1,\dots,d$. Each node $i\in\mathrm{V}$ of $\mathrm{G}$ holds a vector $\xb_i(k)\in\R^b$ evolving at discretized time $t=0,1,2,\dots$. For each clique $\CC_{\mu_l}\in\mathsf{H}_{_{\mathrm{G}}}^\ast$, we assign a matrix $\mathbf{A}_{ij}(\mu_l)\in\R^{b\times b}$ to edge $(i,j)\in\mathrm{E}$ for all $i,j\in\CC_{\mu_l}$ and $\mathbf{A}_{ii}(\mu_l)\in\R^{b\times b}$ to each node $i\in\CC_{\mu_l}$. Introduce a function $\sigma(\cdot):\Z^{\ge 0}\to\{\mu_1,\mu_2,\dots,\mu_d\}$. We define a clique gossip protocol over the graph $\mathrm{G}$ as follows.

\begin{definition}(Clique-gossip Protocol)   At each time $t=0,1,\dots$, one clique $\CC_{\sigma(t)}\in\mathsf{H}_{_{\mathrm{G}}}^\ast$ is selected.  The nodes update their states by
\begin{equation}\notag
\xb_i(t+1)=\left\{
\begin{aligned}
& \sum_{j\in\CC_{\sigma(t)}}\mathbf{A}_{ij}(\sigma(t))\xb_j(t)&\textnormal{ if }i\in\CC_{\sigma(t)};\\
& \xb_i(t)&\textnormal{ if }i\notin\CC_{\sigma(t)}.
\end{aligned}
\right.
\end{equation}
\end{definition}

Note that the signal $\sigma(\cdot)$ plays a role in selecting a clique gossip process which can be deterministic or random.  After $\CC_{\sigma(t)}$ is determined at time $t$, the nodes within the clique $\CC_{\sigma(t)}$ interact with each other as specified by the state transition matrices $\mathbf{A}_{ij}(\mu_l)$. We  remark that at this point we are not imposing any conditions on the  $\mathbf{A}_{ij}(\mu_l)$, whose choices depend on the requirements for individual problems.  For each clique $\CC_{\mu_l}\in\mathsf{H}_{_{\mathrm{G}}}^\ast$, we define a block matrix $\Mb_{\mu_l}\in\R^{nb\times nb}$ whose diagonal blocks equal $\Ib_b$ except the $ii$th block is $\mathbf{A}_{ii}(\mu_l)$ for all $i\in\CC_{\mu_l}$, and off-diagonal blocks equal $\mathbf{0}_b$ except the $ij$th block is $\mathbf{A}_{ij}(\mu_l)$ and the $ji$th block is $\mathbf{A}_{ij}(\mu_l)$ for all $i,j\in\CC_{\mu_l},i\neq j$. Then the above clique-gossip protocol can be put in vector form equivalently
\begin{align}\label{eq:linearsystem}
\xb(t+1) = \Mb_{\sigma(t)} \xb(t),
\end{align}
where $\xb(t)=[\xb_1(t)^\top \dots \xb_n(t)^\top]^\top\in\R^{nb}$.

Therefore, by our definition a clique gossip protocol can be any linear dynamical system that runs over the network $\mathrm{G}$, under which the node interactions take place along a sequence of cliques.  Practically of course we would like the system  ($\ref{eq:linearsystem}$) to asymptotically  converge, preferably to some intrinsically nontrivial limits as solvers to certain network computation problems. This leads us to wonder how we can design  the $\mathbf{A}_{ij}(\mu_l)$ to meet such a criterion in practice. We present the following example as a network linear equation solver \cite{mou2017eigenvalue}.

\medskip

\noindent{\bf Example 1.} Consider a linear algebraic equation with respect to the unknown variable $\yb\in\R^m$
\begin{equation}\label{eq:linearequation}
\Hb\yb = \zb
\end{equation}
with $\Hb\in\R^{n\times m},\zb\in\R^n$. Then (\ref{eq:linearequation}) can be expressed in a system of linear equations $\hb_i^\top\yb=z_i,\ i=1,\dots,n$, where $\hb_i^\top\in\R^m$ denotes the $i$-th row of $\Hb$ and $z_i\in\R$ is the $i$-th component of $\zb$. Assume that (\ref{eq:linearequation}) has a unique solution.

Consider an $n$-node graph $\mathrm{G}=(\mathrm{V},\mathrm{E})$ with a clique coverage $\mathsf{H}_{_{\mathrm{G}}}^\ast$. We let each node $i\in\mathrm{V}$ hold a linear equation $\hb_i^\top\yb=z_i$ and be assigned a state $\xb_i(t),t=0,1,2,\dots$. Suppose each node $i$ is only permitted to share its state with its neighbors. Inspired by the distributed linear equation solver developed in \cite{mou2017eigenvalue} by using the conventional gossip protocol, we apply the clique-gossip protocol to solve (\ref{eq:linearequation}) in a distributed sense as follows. At each time $t=0,1,2,\dots$, we choose $\CC_{\sigma(t)}\in\mathsf{H}_{_{\mathrm{G}}}^\ast$. Then for those nodes $i\notin\CC_{\sigma(t)}$, $\xb_i(t+1)=\xb_i(t)$. For $i\in\CC_{\sigma(t)}$, the update rule is
\begin{align}
\xb_i(t+1) &= \Pb_i(\sum_{j\in\CC_{\sigma(t)}} \xb_j(t)/|\CC_{\sigma(t)}|-\xb_i(t)) + \xb_i(t) \notag\\
&= (\Ib_m-(|\CC_{\sigma(t)}|-1)/|\CC_{\sigma(t)}|\Pb_i)\xb_i(t) + \sum_{j\in\CC_{\sigma(t)},j\neq i}\Pb_i \xb_j(t)/|\CC_{\sigma(t)}|, \label{eq:ex1_updaterule1}
\end{align}
where $\Pb_i=\Ib_m-\hb_i\hb_i^\top/(\hb_i^\top\hb_i)\in\R^{m\times m}$ denotes the projection matrix to the kernel of $\hb_i^\top$. Now we see from (\ref{eq:ex1_updaterule1}) that the solver is an instance of the clique-gossip protocol by letting
\begin{equation}\notag
\Ab_{ij}(\sigma(t))=\left\{
\begin{aligned}
&  \Ib_m-(|\CC_{\sigma(t)}|-1)/|\CC_{\sigma(t)}|\Pb_i&\textnormal{ if }i=j,i,j\in\CC_{\sigma(t)};\\
&\Pb_i/|\CC_{\sigma(t)}|&\textnormal{ if }i\neq j,i,j\in\CC_{\sigma(t)},
\end{aligned}
\right.
\end{equation}
which in turn determines a particular $\Mb_{\sigma(t)}$. One can easily prove that the distributed linear equation solver developed using the clique-gossip protocol drives all nodes of the network to asymptotically compute the solution of (\ref{eq:linearequation}) if the sequence $\CC_{\sigma(0)},$ $\CC_{\sigma(1)},$ $\CC_{\sigma(2)},$ $\dots$ is periodic and the elements in its subsequence over any one period form the clique coverage $\mathsf{H}_{_{\mathrm{G}}}^\ast$.

\subsection{Clique-gossip Averaging Algorithm}\label{sec:clique_gossip_algorithm}

One primary gossip protocol comes from the case where nodes simply average their current states during their meetings, leading to the so-called random or deterministic gossip algorithms. Such gossip algorithms serve as algorithmic descriptions of node interactions over time, and the simple structure of such gossip algorithms enables clear investigation of the convergence rates related to the underlying network structure. Therefore, despite the fact that the exact node interactions can have various different forms in real-world gossip protocols,  the corresponding gossip algorithm is a good indicator to the performance of the protocols.    In the same spirit  now  we  define a clique-gossip averaging algorithm as follows.

\begin{definition}\label{def:algorithm}(Clique-gossip Averaging Algorithm)  Let $\xb_i(t)\in \R$. At time $t$, $\CC_{\sigma(t)}\in\mathsf{H}_{_{\mathrm{G}}}^\ast$ is selected.  The nodes update their states by
\begin{equation}\notag
\xb_i(t+1)=\left\{
\begin{aligned}
& \sum_{j\in\CC_{\sigma(t)}} \xb_j(t)/|\CC_{\sigma(t)}|&\textnormal{ if }i\in\CC_{\sigma(t)};\\
& \xb_i(t)&\textnormal{ if }i\notin\CC_{\sigma(t)}.
\end{aligned}
\right.
\end{equation}
\end{definition}
We can see that the clique-gossip averaging algorithm is an instance of the clique-gossip protocol by setting $\Ab_{ij}(\sigma(t))=1/|\CC_{\sigma(t)}|$ for all $i,j\in\CC_{\sigma(t)}$, which in turn determine $\Mb_{\sigma(t)}$.

\section{Clique Gossip Protocols}\label{sec:invariance}
In this section, we investigate deterministic clique-gossip protocols with periodic clique selections. For the purpose of guaranteeing the reaching of global agreement and the formulation of an eigenvalue invariance theorem, we introduce the following assumption on the function $\sigma(t)$.
\begin{assumption}\label{ass1}
(i) $\sigma(\cdot):\Z^{\ge 0}\to\{\mu_1,\mu_2,\dots,\mu_d\}$ is a periodic function with period $d$; (ii) $\sigma(t)$ visits each element in $\{\mu_1,\mu_2,\dots,\mu_d\}$ once and only once in any period.
\end{assumption}

\subsection{Eigenvalue Invariance}
\subsubsection{Generalized Line Graph}
Recall that given a graph $\mathrm{G}$, its conventional line graph $\mathcal{K}(\mathrm{G})$ is defined by the requirements (i) each node of $\mathcal{K}(\mathrm{G})$ represents an edge of $\mathrm{G}$; (ii) two nodes of $\mathcal{K}(\mathrm{G})$ are linked if and only if the corresponding edges of $\mathrm{G}$ share a common endpoint. In the following, we define the generalized line graph $\mathcal{L}(\mathsf{H}_{_{\mathrm{G}}}^\ast)$ for a graph $\mathrm{G}$ based on the clique coverage $\mathsf{H}_{_{\mathrm{G}}}^\ast$.

\begin{definition}
Let $\mathsf{H}_{_{\mathrm{G}}}^\ast$ be a clique coverage of $\mathrm{G}$. Its generalized line graph, $\mathcal{L}(\mathsf{H}_{_{\mathrm{G}}}^\ast)=(\mathcal{V}(\mathsf{H}_{_{\mathrm{G}}}^\ast),\mathcal{E}(\mathsf{H}_{_{\mathrm{G}}}^\ast))$, is an undirected graph defined by $\mathcal{V}(\mathsf{H}_{_{\mathrm{G}}}^\ast)=\{\CC_i: \CC_i\in\mathsf{H}_{_{\mathrm{G}}}^\ast\}$ and $\mathcal{E}(\mathsf{H}_{_{\mathrm{G}}}^\ast)=\big\{(\CC_i,\CC_j)\in\mathcal{V}(\mathsf{H}_{_{\mathrm{G}}}^\ast)\times\mathcal{V}(\mathsf{H}_{_{\mathrm{G}}}^\ast):\CC_i\cap\CC_j\neq\emptyset,\ i\neq j\big\}$.
\end{definition}
 A sequence $\CC_{i_1},\CC_{i_2},\dots,\CC_{i_k}$ is called a path of cliques if $\CC_{i_j},\CC_{i_{j+1}}$ are adjacent for all $j=1,\dots,k-1$. A cycle of $\mathcal{L}(\mathsf{H}_{_{\mathrm{G}}}^\ast)$ is a path of cliques $\CC_{i_1},\CC_{i_2},\dots,\CC_{i_k}$ such that $\CC_{i_j}\in\mathsf{H}_{_{\mathrm{G}}}^\ast$ for all $j=1,\dots,k$ and $\CC_{i_1}=\CC_{i_k}$.

\begin{figure}[h]
\centering
\includegraphics[width=3.6in]{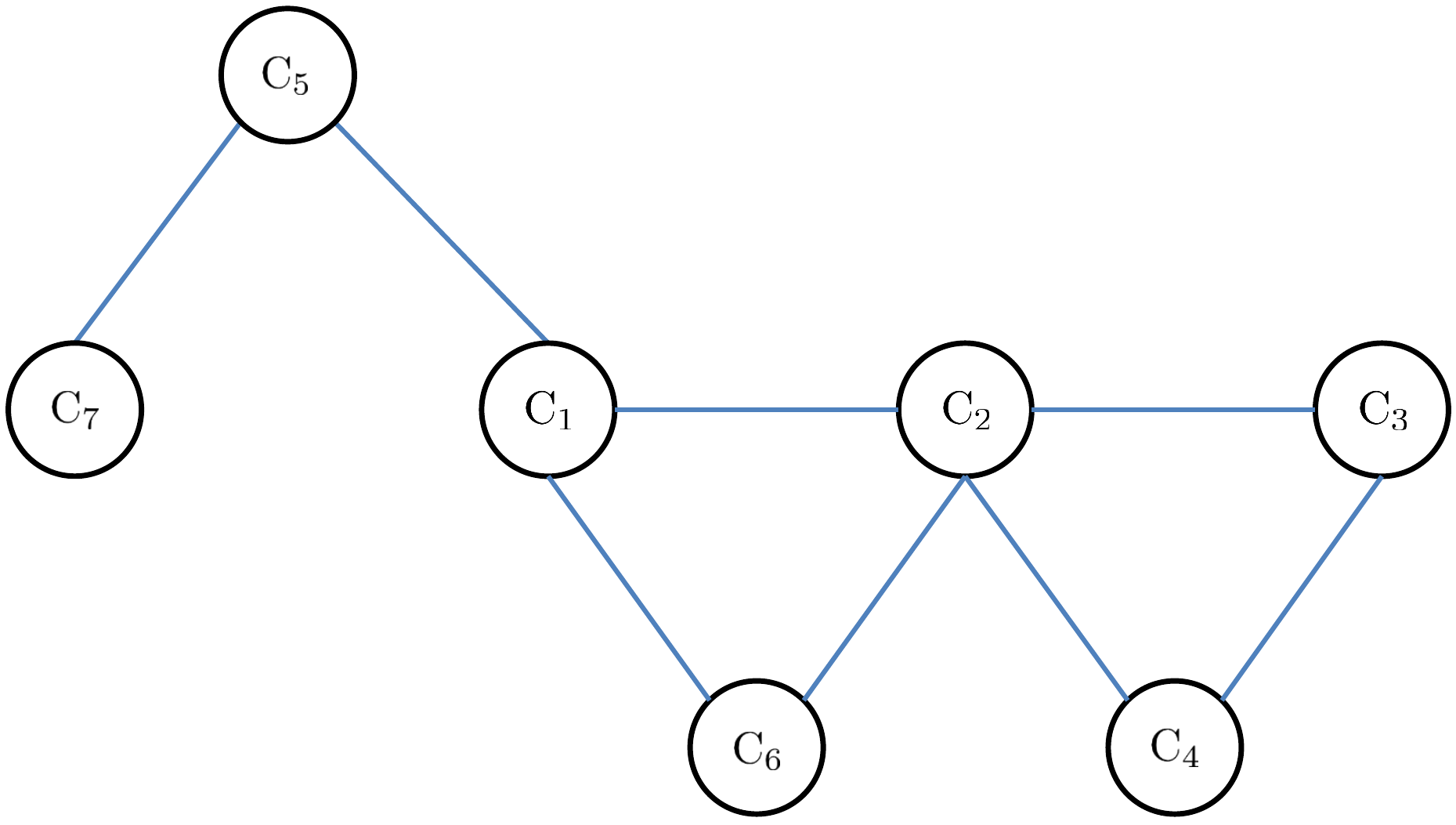}
\caption{The generalized line graph $\mathcal{L}(\mathsf{H}_{_{\mathrm{G}}}^\ast)$ for $\mathrm{G}$ given in Figure \ref{fig:ori}.}
\label{fig:line}
\end{figure}

Note that the generalized line graph is equivalent to the conventional line graph if every clique in the clique coverage contains two nodes. An illustration of a connected graph $\mathrm{G}$ is shown in Figure \ref{fig:ori}, with its generalized line graph given in Figure \ref{fig:line}.
\begin{lemma}
Let $\mathrm{G}$ be a connected graph  with a clique coverage $\mathsf{H}_{_{\mathrm{G}}}^\ast$. Then $\mathcal{L}(\mathsf{H}_{_{\mathrm{G}}}^\ast)$ is a connected graph.
\end{lemma}
\begin{proof}
For two arbitrary cliques $\CC_u,\CC_v\in\mathsf{H}_{_{\mathrm{G}}}^\ast$, we select two nodes $k_u\in\CC_u,k_v\in\CC_v$ of $\mathrm{G}$. Then there exists a path in the union graph $\bigcup\limits_{\CC_i\in\mathsf{H}_{_{\mathrm{G}}}^\ast}\mathrm{G}[\CC_i]$, denoted as $(k_u,k_1),(k_1,k_2),\dots,(k_{n-1},k_n),(k_n,k_v)$, that connects $k_u$ and $k_v$. For the sake of convenience, we let $k_0=k_u,k_{n+1}=k_v$. As a result, there exists a clique $\CC_{j_m}$ such that $(k_m,k_{m+1})$ is an edge in $\mathrm{G}[\CC_{j_m}]$ for each $m=0,\dots,n$. Therefore, the sequence $\CC_u,\CC_{j_0},\dots,\CC_{j_n},\CC_v$, which can have consecutive repeated elements, is a path of cliques that connects $\CC_u$ and $\CC_v$. This completes the proof.
\end{proof}

\subsubsection{A Spectrum Invariance Theorem}

Introduce $\Fb=\Mb_{\sigma(d)}\dots\Mb_{\sigma(1)}$ (Blocks $\Ab_{ij}(\sigma(t)),i,j\in\sigma(t)$ in $\Mb_{\sigma(t)},t=1,\dots,d$ are arbitrary). Then $\Fb$ is the state transition matrix for the periodic gossiping protocol defined by a periodic signal $\sigma(\cdot)$ with period $d$. Let $\pi(\cdot)$ be a permutation with order $d$, i.e., $\pi(\cdot)$ is a one-to-one mapping from $\{1,\dots,d\}$ to $\{1,\dots,d\}$. Denote $\Fb_{\pi}=\Mb_{\sigma(\pi(d))}\dots\Mb_{\sigma(\pi(1))}$. This represents the state transition matrix generated by a permuted order of clique selections. Let $\cp(\mathbf{Q})$ denote the characteristic polynomial for a matrix $\mathbf{Q}$.

Define $\pi_s,s=1,\dots,d-1$ as the swapping permutation over $\{1,\dots,d\}$ with $\pi_s(s)=s+1,\pi_s(s+1)=s$, and $\pi_s(i)=i,i\neq s,i\neq s+1$. In the following, we present a theorem regarding the eigenvalue invariance of the state transition matrix under swapping permutations, generalizing the result of \cite{mou2017eigenvalue}.

\begin{theorem}\label{thm:1}
Let $\mathsf{H}_{_{\mathrm{G}}}^\ast=\{\CC_{\mu_1},\dots,\CC_{\mu_d}\}$ be a clique coverage of $\mathrm{G}=(\mathrm{V},\mathrm{E})$ and let Assumption \ref{ass1} hold. Then along any periodic clique-gossip protocol there holds  $\cp(\Fb) = \cp(\Fb_{\pi_s})$ if $s$ satisfies one of the following conditions:
\begin{enumerate}[(i)]
\item $\CC_{\sigma(s)}$ and $\CC_{\sigma(s+1)}$ are not adjacent,
\item $\CC_{\sigma(s)}$ and $\CC_{\sigma(s+1)}$ are adjacent but neither of them is contained in any cycles of $\mathcal{L}(\mathsf{H}_{_{\mathrm{G}}}^\ast)$.
\end{enumerate}

\end{theorem}

\begin{proof}
It is evident for Condition (i) that $\Mb_{\sigma(s)}\Mb_{\sigma(s+1)}=\Mb_{\sigma(s+1)}\Mb_{\sigma(s)}$ if $\CC_{\sigma(s)}$ and $\CC_{\sigma(s+1)}$ are not adjacent. Thus, in the rest of the proof, we focus on proving $\cp(\Fb)=\cp(\Fb_{\pi_s})$ for $s$ satisfying Condition (ii), i.e.,
\[
\Fb = \Mb_{\sigma(d)}\dots\Mb_{\sigma(s+2)}\Mb_{\sigma(s+1)}\Mb_{\sigma(s)}\dots\Mb_{\sigma(1)},
\]
\[
\Fb_{\pi_s} = \Mb_{\sigma(d)}\dots\Mb_{\sigma(s+2)}\Mb_{\sigma(s)}\Mb_{\sigma(s+1)}\dots\Mb_{\sigma(1)}.
\]
Now we take three steps to complete the proof.

\noindent Step 1. Since $\cp(\Ab\Bb)=\cp(\Bb\Ab)$ for any $\Ab,\Bb\in\R^{nb\times nb}$ (Theorem 1.3.22 \cite{horn2012matrix}), we have
\begin{equation}\label{eq:reorder}
\cp(\Fb) = \cp(\Mb_{\sigma(s+1)}\Mb_{\sigma(s)}\Mb_{\sigma(s-1)}\dots\Mb_{\sigma(1)}\Mb_{\sigma(d)}\dots\Mb_{\sigma(s+2)}).
\end{equation}

\noindent Step 2. In this step, we reorganize the terms in the product $\Mb_{\sigma(s-1)}\dots\Mb_{\sigma(1)}\Mb_{\sigma(d)}\dots\Mb_{\sigma(s+2)}$ by repeatedly interchanging the two consecutive commutable terms. Denote $$\mathrm{A}_1=\{j:\textrm{there exists a path of cliques between }\CC_{\sigma(j)}\textrm{ and }\CC_{\sigma(s)}\textrm{ that does not pass through }\CC_{\sigma(s+1)}\},$$ $$\mathrm{A}_2=\{k:\textrm{there exists a path of cliques between }\CC_{\sigma(k)}\textrm{ and }\CC_{\sigma(s+1)}\textrm{ that does not pass through }\CC_{\sigma(s)}\}.$$
Based on Condition (ii), there hold
\begin{enumerate}[(i)]
\item $\mathrm{A}_1\mcup\mathrm{A}_2=\{s-1,\dots,1,d,\dots,s+2\}$;
\item $\mathrm{A}_1\mcap\mathrm{A}_2=\emptyset$;
\item $\CC_j\mcap\CC_k=\emptyset$ for all $j\in\mathrm{A}_1,k\in\mathrm{A}_2$.
\end{enumerate}
Therefore, $\Mb_{\sigma(j)}\Mb_{\sigma(k)}=\Mb_{\sigma(k)}\Mb_{\sigma(j)}$ if $j\in\mathrm{A}_1,k\in\mathrm{A}_2$ and $j,k$ are two consecutive entries in the sequence $s-1,\dots,1,d,\dots,s+2$. Then it follows
\begin{equation}\label{eq:MMM_P2P1}
\Mb_{\sigma(s-1)}\dots\Mb_{\sigma(1)}\Mb_{\sigma(d)}\dots\Mb_{\sigma(s+2)} = \Pb_2\Pb_1 = \Pb_1\Pb_2,
\end{equation}
where
\[\Pb_1=\Mb_{\sigma(j_{|\mathrm{A}_1|})}\dots\Mb_{\sigma(j_1)},
\]
\[
\Pb_2=\Mb_{\sigma(k_{|\mathrm{A}_2|})}\dots\Mb_{\sigma(k_1)},
\]
with $j_{|\mathrm{A}_1|},\dots,j_1\in\mathrm{A}_1,k_{|\mathrm{A}_2|},\dots,k_1\in\mathrm{A}_2$ following the same order as they are in the sequence $s-1,\dots,1,$ $d,\dots,s+2$. Plugging (\ref{eq:MMM_P2P1}) into (\ref{eq:reorder}), we obtain
\begin{equation}\label{eq:MMP2P1}
\cp(\Fb) = \cp(\Mb_{\sigma(s+1)}\Mb_{\sigma(s)}\Pb_2\Pb_1).
\end{equation}

\noindent Step 3. In this step, we prove $\cp(\Fb)=\cp(\Fb_{\pi})$ and complete the proof. We observe that $\Mb_{\sigma(s)}\Pb_2=\Pb_2\Mb_{\sigma(s)}$, due to the fact that $\CC_{\sigma(s)}$ and $\CC_{\sigma(k)}$ are not adjacent for all $k\in\mathrm{A}_2$. Then (\ref{eq:MMP2P1}) yields
\begin{equation}\label{eq:MP2MP1}
\cp(\Fb) = \cp(\Mb_{\sigma(s+1)}\Pb_2\Mb_{\sigma(s)}\Pb_1),
\end{equation}
which in turn gives
\begin{equation}\label{eq:MP1MP2}
\cp(\Fb) = \cp(\Mb_{\sigma(s)}\Pb_1\Mb_{\sigma(s+1)}\Pb_2).
\end{equation}
Similarly, we also know $\Mb_{\sigma(s+1)}\Pb_1=\Pb_1\Mb_{\sigma(s+1)}$ because $\CC_{\sigma(s+1)}$ and $\CC_{\sigma(j)}$ are not adjacent for all $j\in\mathrm{A}_1$. Finally, we have
\begin{equation}
\begin{aligned}\notag
\cp(\Fb)
&\stackrel{a)}{=}\cp(\Mb_{\sigma(s)}\Mb_{\sigma(s+1)}\Pb_1\Pb_2)\\
&\stackrel{b)}{=}\cp(\Mb_{\sigma(s)}\Mb_{\sigma(s+1)}\Mb_{\sigma(s-1)}\dots\Mb_{\sigma(1)}\Mb_{\sigma(d)}\dots\Mb_{\sigma(s+2)})\\
&\stackrel{c)}{=}\cp(\Mb_{\sigma(d)}\dots\Mb_{\sigma(s+2)}\Mb_{\sigma(s)}\Mb_{\sigma(s+1)}\dots\Mb_{\sigma(1)})\\
&=\cp(\Fb_{\pi_s}),\\
\end{aligned}
\end{equation}
where $a)$ follows from (\ref{eq:MP1MP2}), $b)$ is acquired based on (\ref{eq:MMM_P2P1}), and $c)$ is again due to the fact that $\cp(\Ab\Bb)=\cp(\Bb\Ab)$ for any $\Ab,\Bb\in\R^{nb\times nb}$. This completes the proof.
\end{proof}

\begin{corollary}
Let Assumption \ref{ass1} hold. Then along any clique-gossip protocol there holds there  holds $\cp(\Fb)=\cp(\Fb_\pi)$ for any permutation $\pi$ if the generalized line graph $\mathcal{L}(\mathsf{H}_{_{\mathrm{G}}}^\ast)$ contains no cycle.
\end{corollary}
\begin{proof}
From \cite{fraleigh2003first} we know that an arbitrary permutation from $\{1,\dots,d\}$ to $\{1,\dots,d\}$ can be generated by $d-1$ swapping permutations $\pi_1,\dots,\pi_{d-1}$. Therefore, we only need to prove $\cp(\Fb)=\cp(\Fb_{\pi_s})$ for all $s=1,\dots,d-1$. Based on the condition of the corollary that $\mathcal{L}(\mathsf{H}_{_{\mathrm{G}}}^\ast)$ contains no cycle, either of the two conditions in Theorem \ref{thm:1} is met for all $s$. Hence it can be concluded that $\cp(\Fb)=\cp(\Fb_{\pi_s})$ for all $s$. This completes the proof.
\end{proof}

\subsection{Performance Analysis}
In this section, we analyze the convergence performance of periodic clique-gossip protocols. A precise  definition of the convergence of  a clique-gossip protocol is given below.
\begin{definition}\label{def:convergence}
A clique-gossip protocol is   convergent if there holds
 $$
 \lim\limits_{t\to\infty}\xb(t)=\bar{\yb}(\xb(0))
 $$ for all $\xb(0)\in\R^{nb}$, where $\bar{\yb}(\xb(0))$ is a static state depending perhaps  on $\xb(0)$.
\end{definition}
Let $\Mb_{\sigma(t)},t=0,1,2,\dots$ define a clique-gossip protocol in the form of (\ref{eq:linearsystem}). For a periodic clique-gossip protocol with period $d\in\Z^+$,  we term\footnote{Note that $\Fb=\Mb_{\sigma(d)}\dots\Mb_{\sigma(1)}$ and $\Fb_d=\Mb_{\sigma(d-1)}\dots\Mb_{\sigma(0)}$ have the same spectrum regardless of the underlying network structure and the generalized line graph of the clique coverage.  }
\[
\Fb_d=\Mb_{\sigma(d-1)}\dots\Mb_{\sigma(0)},
\]
as a period-based state transition matrix in view of the recursion
\[
\xb((q+1)d)=\Fb_d\xb(qd),\ q=0,1,2,\dots.
\]

Let $\sigma(\mathbf{A})$ and $\rho(\mathbf{A})$ denote the spectrum and spectral radius of a matrix $A$, respectively. The following lemma holds from  the basic knowledge of linear systems.
\begin{lemma}\label{lem:convergence}
Let Assumption 1 hold. Let the periodic clique-gossip protocol admit a period-based state transition matrix $\Fb_d\in\R^{nb\times nb}$. The protocol is convergent if and only if the following conditions hold:
\begin{enumerate}[(i)]
\item $\rho(\Fb_d)\le 1$;
\item If $1\in\sigma(\Fb_d)$, then eigenvalue one has equal algebraic multiplicity and geometric multiplicity;
\item If $\lambda\in\sigma(\Fb_d)$ and $|\lambda|=1$, then $\lambda=1$;
\item There holds $\Mb_{\sigma(k)}\dots\Mb_{\sigma(0)} \betab =\betab$ for all $\betab\in\mathbb{I}:=\{\alphab \in \mathbb{R}^{nb}:\Fb_d\alphab=\alphab\}$.
\end{enumerate}
\end{lemma}

Next, we define precisely the convergence rates of   standard gossiping and clique-gossiping protocols.
\begin{definition}\label{def:convergence_rate}
For a convergent  clique-gossip protocol,  the rate of  convergence is $\mathcal{O}(\nu^t)$ if there exists a unique $\nu\in(0,1)$ such that
\begin{equation}\notag
0<\limsup\limits_{t\to\infty} \frac{\|\xb(t)-\bar{\yb}\|}{\nu^t}<+\infty,\ \forall\xb(0)\notin\mathbb{I},
\end{equation}
where $\bar{\yb}=\lim\limits_{t\to\infty}\xb(t)$.
\end{definition}

 Introduce the symbol $|\lambda_2(\Fb)|$ as the magnitude of the eigenvalues of $\Fb\in\R^{nb\times nb}$ with the second largest modulus of all its eigenvalues, i.e., $|\lambda_2(\Fb)|=\max\{|\lambda|:\lambda\in\sigma(\Fb),|\lambda|<\rho(\Fb)\}$.
Let $\ceil{x}$ be the smallest integer greater than or equal to $x\in\R$ and $\floor{x}$ be the largest integer less than or equal to $x\in\R$. Now we present a proposition that reveals the relationship between $|\lambda_2(\Fb_d)|$ and the convergence rate $\nu$ for clique-gossping.
\begin{proposition}\label{prop:convergence_rate}
Let Assumption 1 hold and consider the resulting  periodic clique-gossip protocol with period $d\in\Z^+$. Let $\Fb_d\in\R^{nb\times nb}$ be a period-based state transition matrix of the protocol and assume the protocol is convergent. Suppose $|\lambda_2(\Fb_d)|>0$. Then there holds $
\nu=|\lambda_2(\Fb_d)|^{1/d}$, i.e., the  convergence rate of the protocol  is $\mathcal{O}\big(|\lambda_2(\Fb_d)|^{t/d}\big)$.
\end{proposition}
\begin{proof}
By the basic knowledge of the stability of linear systems
\begin{equation}\label{eq:rate1}
\limsup\limits_{t\to\infty}\frac{\|\xb(\floor{\frac{t}{d}}d)-\bar{\yb}\|}{|\lambda_2(\Fb_d)|^{\floor{\frac{t}{d}}}}=\limsup\limits_{t\to\infty}\frac{\|\xb(\ceil{\frac{t}{d}}d)-\bar{\yb}\|}{|\lambda_2(\Fb_d)|^{\ceil{\frac{t}{d}}}}=C(\xb(0)).
\end{equation}
with $\bar{\yb}=\lim\limits_{t\to\infty}\xb(t)$, where $C(\xb(0))$ is a constant relying on the network initial value. Let $\Mb_{\sigma(t)},t\ge 0$ be the state transition matrix corresponding to each $\CC_{\sigma(t)}$ in the clique gossip sequence. By the convergence of $\xb(t)$, we have
\begin{equation}\label{eq:rate01}
\bar{\yb}=\lim\limits_{t\to\infty}\xb(t)
=\lim\limits_{q\to\infty}\xb((q+1)d)=\lim\limits_{t\to\infty} \mathbf{F}_d\xb(qd)=\mathbf{F}_d\bar{\yb},
\end{equation}
which yields $\bar{\yb}\in\mathbb{I}$ for all $\xb(0)$. By (\ref{eq:rate1}), (\ref{eq:rate01}), and Lemma \ref{lem:convergence}.(iv), we obtain
\begin{align}
\limsup\limits_{t\to\infty}\frac{\|\xb(t)-\bar{\yb}\|}{|\lambda_2(\Fb_d)|^{t/d}}
&=\limsup\limits_{t\to\infty}\frac{\|\Mb_{\sigma(t-1)}\dots\Mb_{\sigma(\floor{\frac{t}{d}}d)}\xb(\floor{\frac{t}{d}}d)-\bar{\yb}\|}{|\lambda_2(\Fb_d)|^{t/d}}\notag\\
&= \limsup\limits_{t\to\infty}\frac{\|\Mb_{\sigma(t-1)}\dots\Mb_{\sigma(\floor{\frac{t}{d}}d)}(\xb(\floor{\frac{t}{d}}d)-\bar{\yb})\|}{|\lambda_2(\Fb_d)|^{t/d}}\notag\\
&\le \|\Mb_{\sigma(t-1)}\dots\Mb_{\sigma(\floor{\frac{t}{d}}d)}\|\limsup\limits_{t\to\infty}\frac{\|\xb(\floor{\frac{t}{d}}d)-\bar{\yb}\|}{|\lambda_2(\Fb_d)|^{t/d}}\notag\\
&\le \|\Mb_{\sigma(t-1)}\dots\Mb_{\sigma(\floor{\frac{t}{d}}d)}\|\limsup\limits_{t\to\infty}\frac{\|\xb(\floor{\frac{t}{d}}d)-\bar{\yb}\|}{|\lambda_2(\Fb_d)|^{\floor{\frac{t}{d}}}}|\lambda_2(\Fb_d)|^{\floor{\frac{t}{d}}-\frac{t}{d}}\notag\\
&<|\lambda_2(\Fb_d)|^{-1}B_1C(\xb(0)),\label{eq:rate2}
\end{align}
with $B_1=\max\{\|\Mb_{\sigma(k)}\dots\Mb_{\sigma(0)}\|:k=0,\dots,d-1\}$. Similarly, noticing
\begin{equation}\notag
\limsup\limits_{t\to\infty}\frac{\|\Mb_{\sigma(\ceil{\frac{t}{d}}d-1)}\dots\Mb_{\sigma(t)}\xb(t)-\bar{\yb}\|}{|\lambda_2(\Fb_d)|^{t/d}}
=\limsup\limits_{t\to\infty}\frac{\|\xb(\ceil{\frac{t}{d}}d)-\bar{\yb}\|}{|\lambda_2(\Fb_d)|^{t/d}},
\end{equation} one also has
\begin{equation}\label{eq:rate3}
\limsup\limits_{t\to\infty}\frac{\|\xb(t)-\bar{\yb}\|}{|\lambda_2(\Fb_d)|^{t/d}} > |\lambda_2(\Fb_d)|B_2^{-1}C(\xb(0)),
\end{equation}
with $B_2=\min\{\|\Mb_{\sigma(d-1)}\dots\Mb_{\sigma(k)}\|:k=0,\dots,d-1\}$. From (\ref{eq:rate2}) and (\ref{eq:rate3}), the desired characterization to the rate of convergence follows.
\end{proof}

We note that the above discussions on convergence and convergence rate of clique-gossip protocols cover standard gossip protocols since standard gossiping is a special case of clique-gossiping with all cliques being node pairs.

\subsection{Multi-clique-gossip Protocols}
Recall that in a clique-gossip protocol, one clique is selected at each time slot and the clique-gossiping operation is undertaken among all nodes in this clique, while the other nodes maintain their states. For the purpose of speeding up the information spreading over networks, we generalize the notion of multi-gossip in a standard gossip process \cite{liu2011deterministic} to a clique gossip. Based on the clique coverage $\mathsf{H}_{_{\mathrm{G}}}^\ast=\{\CC_{\mu_1},\dots,\CC_{\mu_d}\}$, we define a multi-clique coverage $\mathpzc{M}(\mathsf{H}_{_{\mathrm{G}}}^\ast)=\{\mathsf{C}_1,\dots,\mathsf{C}_{\nu}\}$ with each $\mathsf{C}_k\subset\mathsf{H}_{_{\mathrm{G}}}^\ast$, termed a clique class, being a set of $\CC_{\mu_l}$, where $\bigcup_{k=1}^{\nu}\mathsf{C}_k=\mathsf{H}_{_{\mathrm{G}}}^\ast$, and where any two cliques in one clique class $\mathsf{C}_k$ are non-adjacent for all $k=1,\dots,\nu$. We also define $\mathrm{C}_i^\ast(\mathsf{C}_k)\in\mathsf{C}_k$ as the unique clique containing node $i$ and belonging to $\mathsf{C}_k$. Introduce a function $\sigma(\cdot):\Z^{\ge 0}\to\{1,\dots,\nu\}$. Then a multi-clique-gossip protocol is defined as follows.

\medskip

\noindent{\bf Example 2.} Consider the graph $\mathrm{G}$ in Figure \ref{fig:ori}. Let cliques $\CC_1=\{1,2,3\},\CC_2=\{3,4\},\CC_3=\{4,5,6,7\},\CC_4=\{4,6,7,8\},\CC_5=\{2,9,13\},\CC_6=\{3,11,12\},\CC_7=\{9,10\}$ form its clique coverage $\mathsf{H}_{_{\mathrm{G}}}^\ast$. Then its generalized line graph $\mathcal{L}(\mathsf{H}_{_{\mathrm{G}}}^\ast)$ is given in Figure \ref{fig:line}. By observing $\mathcal{L}(\mathsf{H}_{_{\mathrm{G}}}^\ast)$, we can obtain that possible multi-clique coverages $\mathpzc{M}(\mathsf{H}_{_{\mathrm{G}}}^\ast)$ include $\big\{\{\CC_1,\CC_4,\CC_7\},\{\CC_3,\CC_5,\CC_6\},\{\CC_2\}\big\},\big\{\{\CC_1,\CC_7\},\{\CC_2,\CC_5\},\{\CC_4,\CC_6\},\{\CC_3\}\big\}$.

\begin{definition}
(Multi-clique-gossip Protocol) Select $\mathsf{C}_{\sigma(t)}\in\mathpzc{M}(\mathsf{H}_{_{\mathrm{G}}}^\ast)$ at each time $t=0,1,\dots$. Then the node state update rule is described by
\begin{equation}\notag
\xb_i(t+1)=\left\{
\begin{aligned}
& \sum_{j\in\mathrm{C}_i^\ast(\mathsf{C}_{\sigma(t)})}\mathbf{A}_{ij}(\mu_l)\xb_j(t)&\textnormal{ if } i\in\bigcup_{\CC_{\mu_l}\in\mathsf{C}_{\sigma(t)}}\CC_{\mu_l};\\
& \xb_i(t)&\textnormal{ if } i\notin\bigcup_{\CC_{\mu_l}\in\mathsf{C}_{\sigma(t)}}\CC_{\mu_l}.
\end{aligned}
\right.
\end{equation}
\end{definition}
We can see that in contrast to the clique-gossip protocol, the multi-clique-gossip protocol allows multiple \emph{non-adjacent} cliques to perform internal clique-gossip operations simultaneously. Evidently, the simultaneous operations over non-adjacent cliques are not mutually influential because no node serves as the intermediary for information transmission.
By direct intuition, we know that in order to speed up the convergence to a global agreement, one should arrange as many cliques as possible to perform gossiping operation in every time slot, i.e., minimize $|\mathpzc{M}(\mathsf{H}_{_{\mathrm{G}}}^\ast)|$. Define the clique-class index by $\rho(\mathsf{H}_{_{\mathrm{G}}}^\ast)=\min\{|\mathpzc{M}(\mathsf{H}_{_{\mathrm{G}}}^\ast)|:\mathpzc{M}(\mathsf{H}_{_{\mathrm{G}}}^\ast)\textnormal{ is a multi-clique coverage induced by }\mathsf{H}_{_{\mathrm{G}}}^\ast\}$. Let $\Delta(\mathcal{L}(\mathsf{H}_{_{\mathrm{G}}}^\ast))$ denote the maximum node degree of the generalized line graph $\mathcal{L}(\mathsf{H}_{_{\mathrm{G}}}^\ast)$. Define $\alpha(\mathcal{L}(\mathsf{H}_{_{\mathrm{G}}}^\ast))=\max\{|\mathsf{C}|:\mathsf{C}\subset\mathsf{H}_{_{\mathrm{G}}}^\ast | \CC_i\cap\CC_j=\emptyset,\forall\CC_i,\CC_j\in\mathsf{C}\}$ as the independence number of $\mathcal{L}(\mathsf{H}_{_{\mathrm{G}}}^\ast)$. Then we have the following proposition.
\begin{proposition}\label{prop:1}
If $\mathcal{L}(\mathsf{H}_{_{\mathrm{G}}}^\ast)$ is neither a complete graph nor a cycle graph with an odd number of nodes, then
\begin{equation}\notag
\frac{|\mathsf{H}_{_{\mathrm{G}}}^\ast|}{\alpha(\mathcal{L}(\mathsf{H}_{_{\mathrm{G}}}^\ast))}\le\rho(\mathsf{H}_{_{\mathrm{G}}}^\ast)\le\Delta(\mathcal{L}(\mathsf{H}_{_{\mathrm{G}}}^\ast)).
\end{equation}
In particular, $\rho(\mathsf{H}_{_{\mathrm{G}}}^\ast)=|\mathsf{H}_{_{\mathrm{G}}}^\ast|$ if $\mathcal{L}(\mathsf{H}_{_{\mathrm{G}}}^\ast)$ is a complete graph, and $\rho(\mathsf{H}_{_{\mathrm{G}}}^\ast)=3$ if $\mathcal{L}(\mathsf{H}_{_{\mathrm{G}}}^\ast)$ is a cycle graph with an odd number of nodes.
\end{proposition}
\begin{proof}
The left inequality naturally holds because of the definitions of $\alpha(\cdot)$ and $\rho(\cdot)$. Now we prove the right inequality. Consider the vertex coloring problem of the generalized line graph $\mathcal{L}(\mathsf{H}_{_{\mathrm{G}}}^\ast)$, which is a labeling of the graph's nodes with colors such that any two nodes which are the endpoints of some edge have different colors. We denote the smallest number of colors needed to color the nodes of graph $\mathcal{L}(\mathsf{H}_{_{\mathrm{G}}}^\ast)$, namely its chromatic number, as $\chi(\mathcal{L}(\mathsf{H}_{_{\mathrm{G}}}^\ast))$. It can observed that based on the same clique coverage $\mathsf{H}_{_{\mathrm{G}}}^\ast$, the minimum number of clique classes equals the chromatic number of its generalized line graph, i.e., $\rho(\mathsf{H}_{_{\mathrm{G}}}^\ast)=\chi(\mathcal{L}(\mathsf{H}_{_{\mathrm{G}}}^\ast))$. Then by  Brooks' Theorem \cite{brooks1941colouring}, $\chi(\mathrm{G})\le\Delta(\mathrm{G})$ if $\mathrm{G}$ is a simple connected graph but not a complete graph or a cycle graph with odd nodes. Therefore, we have $\rho(\mathsf{H}_{_{\mathrm{G}}}^\ast)\le\Delta(\mathcal{L}(\mathsf{H}_{_{\mathrm{G}}}^\ast))$ unless $\mathcal{L}(\mathsf{H}_{_{\mathrm{G}}}^\ast)$ is a complete graph or a cycle graph with odd nodes. The particular values of $\rho(\mathsf{H}_{_{\mathrm{G}}}^\ast)$ for complete graphs and odd cycle graphs result easily from their specific structures.
\end{proof}
From the proof of Proposition \ref{prop:1}, it can be seen that finding the clique classes of a graph is equivalent to finding a vertex coloring of its generalized line graph. It is known \cite{mou2010deterministic} that finding the number of conventional multigossips of a graph is intrinsically obtaining an edge coloring of the graph, which is a labeling of the edges of the graph such that any two edges sharing the same endpoint have different colors. Then it follows that the edge coloring of the graph is equivalent to a vertex coloring of its conventional line graph. Since the conventional line graph is a special case of the generalized line graph by letting every clique in $\mathsf{H}_{_{\mathrm{G}}}^\ast$ possess two nodes, we can conclude that the problem of finding the clique classes of a graph in this paper is consistent with the result regarding multigossips in \cite{mou2010deterministic}.

It is hard to find $\rho(\mathsf{H}_{_{\mathrm{G}}}^\ast)$ of an arbitrary graph $\mathrm{G}$. Inspired by the greedy algorithm in \cite{cormen1990introduction}, however, we can generate a multi-clique coverage $\mathpzc{M}(\mathsf{H}_{_{\mathrm{G}}}^\ast)$ from the clique coverage $\mathsf{H}_{_{\mathrm{G}}}^\ast$ by visiting every node of $\mathcal{L}(\mathsf{H}_{_{\mathrm{G}}}^\ast)$ in order and assign it into the first available clique class, so that we can obtain a relatively small $|\mathpzc{M}(\mathsf{H}_{_{\mathrm{G}}}^\ast)|$.

\subsection{Numerical Examples}
In this section, we provide a few numerical examples to illustrate the result in Theorem \ref{thm:1} and investigate the performance of the clique-gossip averaging algorithm by comparing it to the standard gossip algorithm, and the multi-clique-gossiping in contrast to pure clique-gossiping.

\subsubsection{Validation of Theorem \ref{thm:1}}
The following example validates the result in Theorem \ref{thm:1}.

\noindent {\bf Example 3.} Consider the graph $\mathrm{G}$ in Figure \ref{fig:ori} with the clique coverage $\mathsf{H}_{_{\mathrm{G}}}^\ast=\{\CC_1,\dots,\CC_7\}$, where $\CC_1,\dots,\CC_7$ are specified in Figure \ref{fig:ori}.  In order to validate two conditions in Theorem \ref{thm:1}, we compute the spectrum of
$$\Fb = \Mb_1\Mb_4\Mb_5\Mb_7\Mb_2\Mb_3\Mb_6,$$
$$\Fb_{\pi_6}=\Mb_1\Mb_4\Mb_5\Mb_7\Mb_2\Mb_6\Mb_3,$$
$$\Fb_{\pi_3}=\Mb_1\Mb_4\Mb_7\Mb_5\Mb_2\Mb_3\Mb_6,$$
where $\Mb_{\mu},\mu=1,\dots,7$, corresponding to $\CC_{\mu},\mu=1,\dots,7$, are as defined as in Section \ref{sec:clique_gossip_algorithm}. This implies that $\Fb,\Fb_{\pi_6},\Fb_{\pi_3}$ are the state transition matrices for the clique-gossip averaging algorithm. Obviously $\pi_6$ is the permutation that interchanges two non-adjacent cliques $\CC_3,\CC_6$, and $\pi_3$ interchanges $\CC_5,\CC_7$, neither of which is contained in any cycle of $\mathcal{L}(\mathsf{H}_{_{\mathrm{G}}}^\ast)$ plotted in Figure \ref{fig:line}.  As computed, $\sigma(\Fb)=\sigma(\Fb_{\pi_6})=\sigma(\Fb_{\pi_3})=\{1,0.8504,0.6920,0.3683,0.1522,0.1871,0,0,0,0,0,0,0\}$, i.e., $\cp(\Fb)=\cp(\Fb_{\pi_6})=\cp(\Fb_{\pi_3})$. This is consistent with Theorem \ref{thm:1}.

\subsubsection{Performance Discussion}

In the following example, we compare the convergence speed of the periodic clique-gossip averaging algorithm and the standard periodic gossip averaging algorithm and discuss the performance improvement with the application of multi-clique-gossiping.

\noindent {\bf Example 4.} Consider the graph $\mathrm{G}$ in Figure \ref{fig:ori}. Let $\Mb_{\mu},\mu=1,\dots,7$ be the same as in Example 2. Denote the system states along standard gossiping and clique-gossiping as $\xb_g(t),\xb_c(t),t=0,1,2,\dots$, respectively. Define
\begin{equation}\notag
\Fb_c = \Mb_5\Mb_7\Mb_1\Mb_6\Mb_2\Mb_3\Mb_4
\end{equation}
as the period-based state transition matrix for the periodic clique-gossip averaging algorithm. Then we let $(2,13),$$(2,9),$$(9,13),$$(9,10),$$(1,2),$$(1,3),$$(2,3),$$(3,12),$$(3,11),$$(11,13),$$(3,4),$$(8,11),$$(4,5),$$(5,6),$$(4,6),$$(6,7),$     $(7,8),$     $(4,8),$$(4,7),(6,8)$ be a $20$ edges gossip sequence for the standard periodic gossip averaging algorithm. Denote $\Fb_g$ as the period-based state transition matrix corresponding to the standard gossip sequence. First we compute that the second largest eigenvalues of $\Fb_g,\Fb_c$ are $0.8061,0.8504$, respectively. Then by Proposition \ref{prop:convergence_rate}, the convergence speed for these two algorithms can be represented by $\sqrt[20]{0.8061},\sqrt[7]{0.8504}$, respectively.

Next we plot the trajectories of error $e(t)=\sum\limits_{i=1}^7 (\xb_i(t)-\bar{\xb})^2$, with $\bar{\xb}=\sum\limits_{i=1}^7 \xb_i(0)/7$, for these two algorithms in Figure \ref{fig:example3}. It is known \cite{liu2011deterministic} that the second largest eigenvalue of the state transition matrix determines the convergence speed of its corresponding algorithm. Then we can conclude from the fact that $\sqrt[7]{0.8504}<\sqrt[20]{0.8061}$ that the periodic clique-gossip averaging algorithm has faster convergence speed than the standard periodic gossip averaging algorithm. Moreover, this conclusion is directly verified as shown in Figure \ref{fig:example3}, because it can be clearly seen that the trajectories of error for clique-gossiping are steeper than gossiping, implying that the states of all nodes of $\mathrm{G}$ undertaking clique-gossiping approach the average $\bar{\xb}$ faster.

Based on the clique coverage $\mathsf{H}_{_{\mathrm{G}}}^\ast=\{\CC_1,\dots,\CC_7\}$ and the generalized line graph $\mathcal{L}(\mathsf{H}_{_{\mathrm{G}}}^\ast)$ in Figure \ref{fig:line}, we define a multi-clique coverage $\mathpzc{M}(\mathsf{H}_{_{\mathrm{G}}}^\ast)=\{\mathsf{C}_1,\mathsf{C}_2,\mathsf{C}_3\}$ with $\mathsf{C}_1=\{\CC_1,\CC_4,\CC_7\},\mathsf{C}_2=\{\CC_3,\CC_5,\CC_6\},\mathsf{C}_3=\{\CC_2\}$. Let the multi-clique gossiping occur over the multi-clique coverage $\mathpzc{M}(\mathsf{H}_{_{\mathrm{G}}}^\ast)$. Then the trajectory of error $e(t)$ for multi-clique-gossiping is plotted in Figure \ref{fig:example3}. It can be seen that multi-clique-gossiping yields much faster convergence speed than either clique-gossiping or standard gossiping. To make this conclusion numerically clear, we calculate that the second largest eigenvalue of the state transition matrix for the multi-clique-gossiping $\sqrt[3]{0.8504}$ is less than that for the clique-gossiping $\sqrt[7]{0.8504}$.

\begin{figure}[H]
\centering
\includegraphics[width=4in]{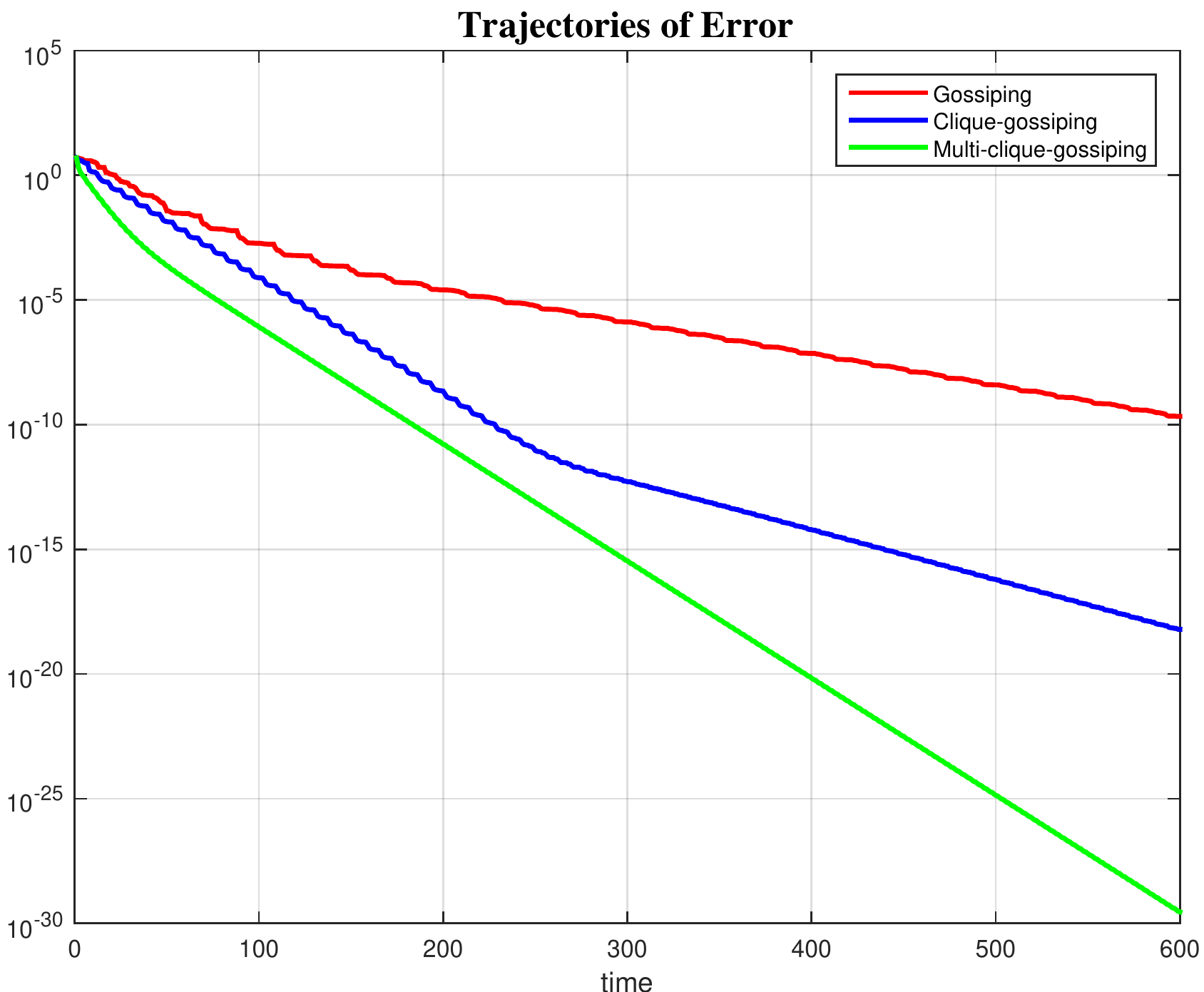}
\caption{The trajectories of error $e(t)=\sum\limits_{i=1}^7 (\xb_i(t)-\bar{\xb})^2$ with $\bar{\xb}=\sum\limits_{i=1}^7 \xb_i(0)/7$ for the periodic multi-clique-gossiping, clique-gossip averaging algorithm and the standard periodic gossip averaging algorithm.}
\label{fig:example3}
\end{figure}

\medskip

Now we study a few typical graphs and provide two examples to illustrate the way that clique-gossiping yields faster convergence speed than standard gossiping.

\noindent{\bf Example 5.} Consider the graphs $\mathrm{G}_m$ with $m=3,\dots,20$ whose topologies are given in Figure \ref{fig:twin_cycle}. It can be seen that $\mathrm{G}_m$ has a typical structure that the induced graphs $\mathrm{G}[\{1,2,\dots,m\}]$ and $\mathrm{G}[\{m+2,m+3,\dots,2m+1\}]$ are both ring graphs, which are linked by a complete induced graph $\mathrm{G}[\{1,m+1,m+2\}]$. Then it follows that $\CC_1=\{1,m+1,m+2\}$ is a $3$-node clique. Let $e_1=(1,m+1),e_2=(1,m+2),e_3=(m+1,m+2)$. Define clique sequence $\mathrm{S}_{\rm leftp}=(1,2),(2,3),\dots,(m,1),\mathrm{S}_{\rm right}=(m+2,m+3),(m+3,m+4),\dots,(2m+1,m+2)$. Let $\mathrm{S}_{\rm left},e_3,e_1,e_2,\mathrm{S}_{\rm  right}$ and $\mathrm{S}_{\rm  left},e_1,e_2,e_3,\mathrm{S}_{\rm  right}$ be two standard gossip averaging sequences with their period-based state transition matrix denoted by $\Fb_{g_1},\Fb_{g_2}$, respectively. Correspondingly, we replace $(1,m+1),(m+1,m+2),(1,m+2)$ with $\CC_1$ to form the clique-gossip averaging sequence, with its period-based state transition matrix denoted by $\Fb_c$. Note that the period length for clique-gossiping is shorter than that for standard gossiping. First we plot the values of $|\lambda_2(\Fb_{g_1})|,|\lambda_2(\Fb_{g_2})|,|\lambda_2(\Fb_c)|$ varying with $m=3,\dots,20$ in Figure \ref{fig:lambda2_twin_cycle}. As can be seen, $|\lambda_2(\Fb_{g_2})|<|\lambda_2(\Fb_{g_1})|=|\lambda_2(\Fb_c)|$ for all $m$. This shows that the application of the clique gossiping does not necessarily  reduce the second largest eigenvalue of the period-based state transition matrix. Based on Proposition \ref{prop:convergence_rate}, we next investigate the relationship among the convergence rate $|\lambda_2(\Fb_{g_1})|^{1/(2m+3)},|\lambda_2(\Fb_{g_2})|^{1/(2m+3)}$ and $|\lambda_2(\Fb_c)|^{1/(2m+1)}$ for all $m=3,\dots,20$ in Figure \ref{fig:lambda2_twin_cycle_nthroot}. The calculated result shows $|\lambda_2(\Fb_c)|^{1/(2m+1)}<|\lambda_2(\Fb_{g_2})|^{1/(2m+3)}<|\lambda_2(\Fb_{g_1})|^{1/(2m+3)}$ for all $m=3,\dots,20$, which indicates that clique-gossiping has faster convergence speed than standard gossiping, especially when the ring graphs $\mathrm{G}[\{1,2,\dots,m\}]$ and $\mathrm{G}[\{m+2,m+3,\dots,2m+1\}]$ have small size.

\begin{figure}[h]
\centering
\includegraphics[width=4.2in]{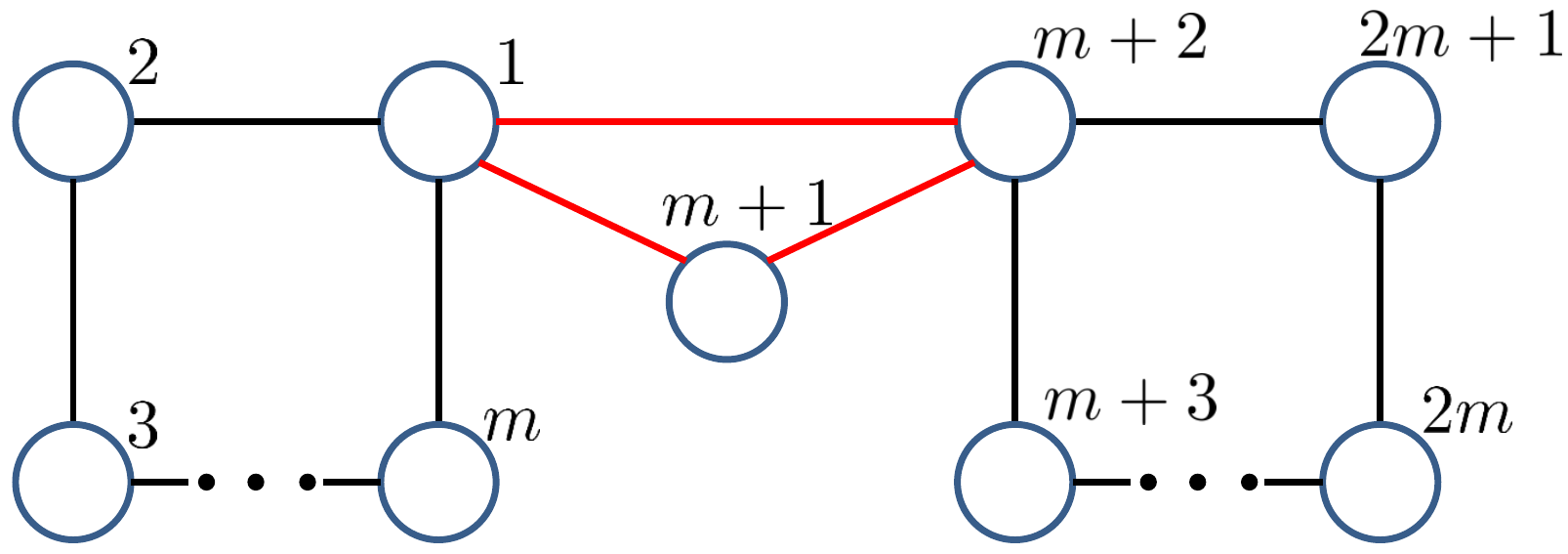}
\caption{Graph $\mathrm{G}_m,m=3,\dots,20$.}
\label{fig:twin_cycle}
\end{figure}

\begin{figure}[h]
\centering
\includegraphics[width=4in]{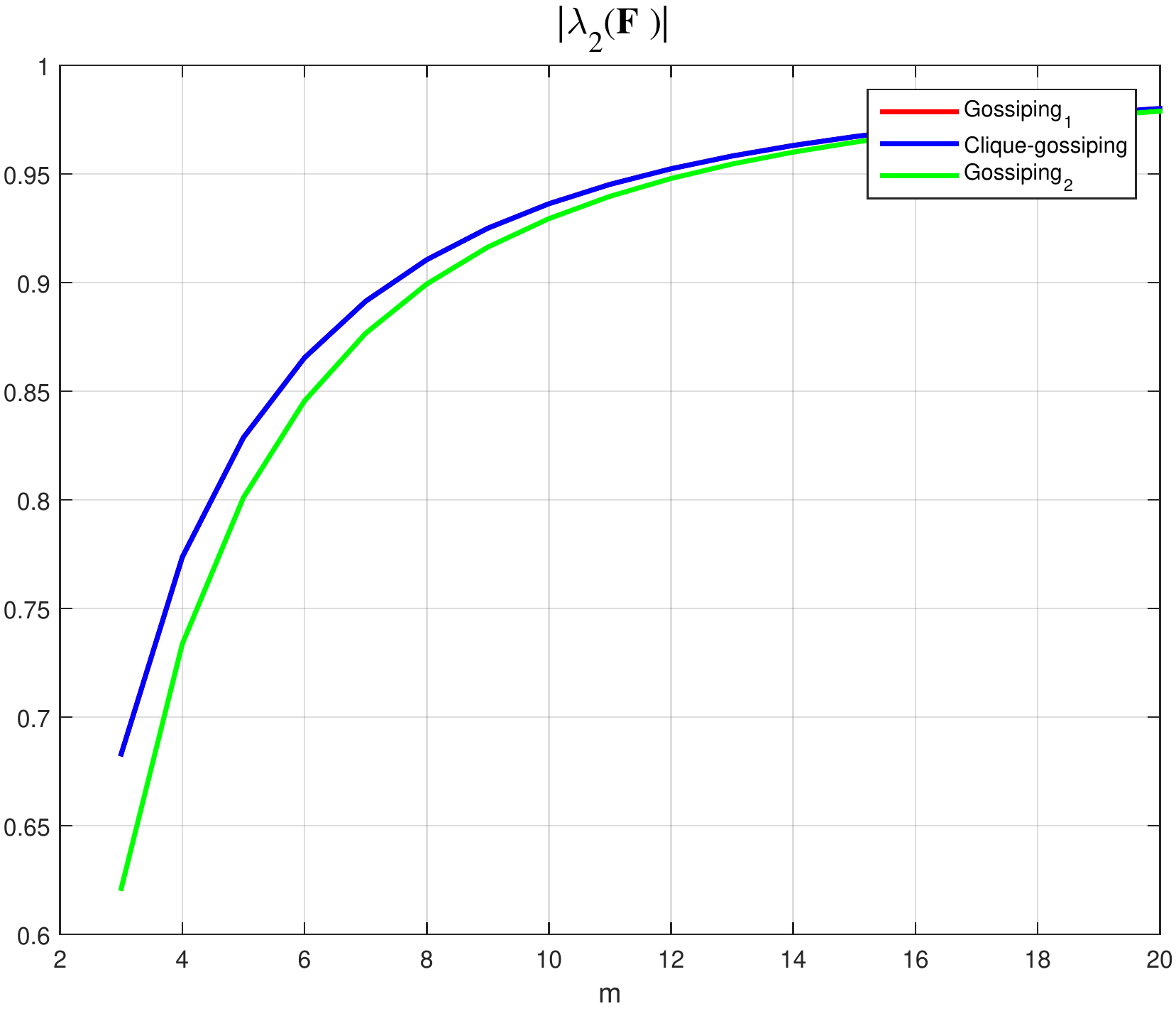}
\caption{The values of $|\lambda_2(\Fb_{g_1})|,|\lambda_2(\Fb_{g_2})|$ for standard gossiping and $|\lambda_2(\Fb_c)|$ for clique-gossiping varying with $m=3,\dots,20$.}
\label{fig:lambda2_twin_cycle}
\end{figure}

\begin{figure}[h]
\centering
\includegraphics[width=4in]{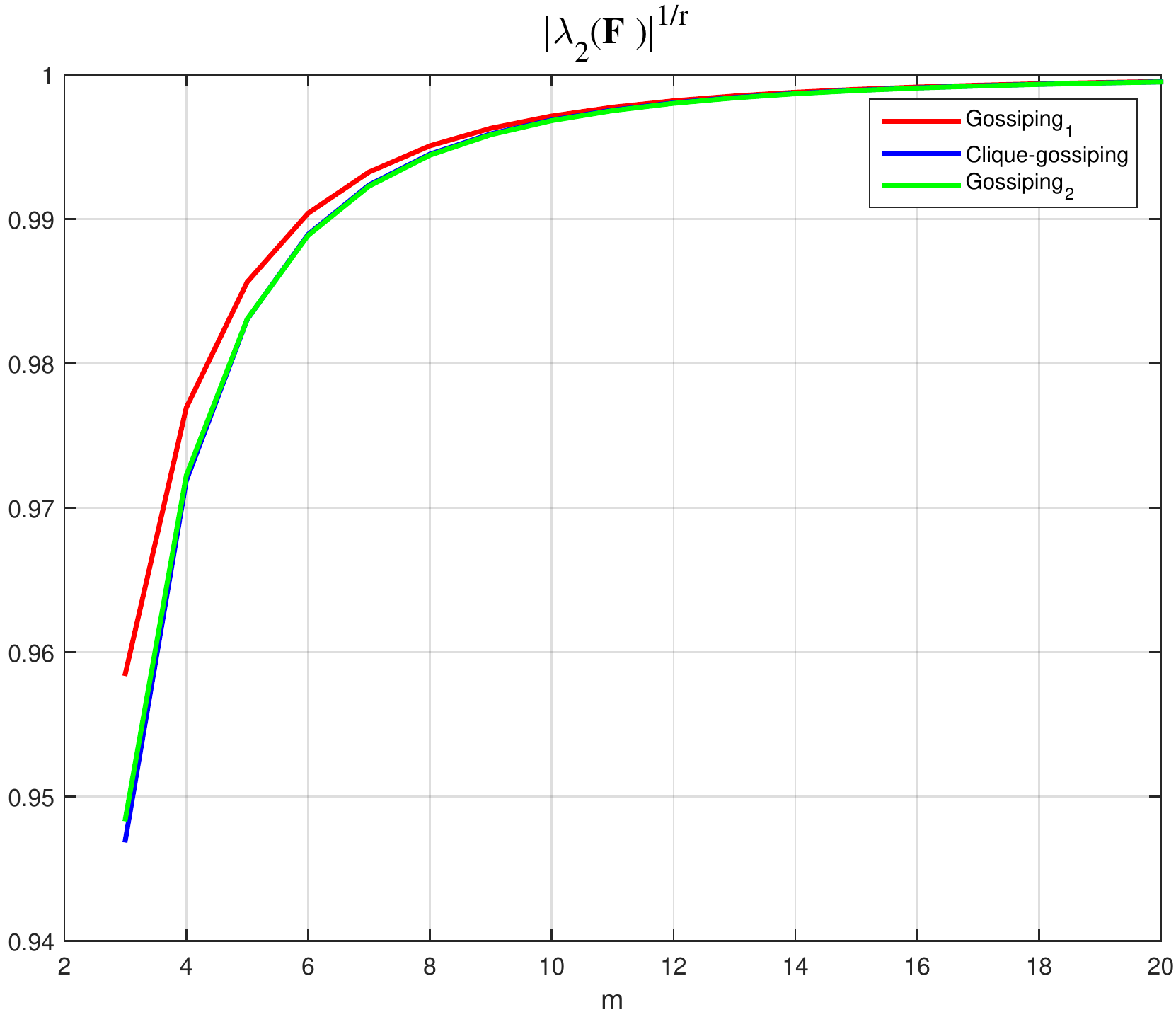}
\caption{The trajectories of $|\lambda_2(\Fb_{g_1})|^{1/(2m+3)},|\lambda_2(\Fb_{g_2})|^{1/(2m+3)}$ for standard gossiping and $|\lambda_2(\Fb_c)|^{1/(2m+1)}$ for clique-gossiping varying with $m=3,\dots,20$.}
\label{fig:lambda2_twin_cycle_nthroot}
\end{figure}

\noindent{\bf Example 6.} Consider the $101$-node graphs $\mathrm{G}_k=(\mathrm{V},\mathrm{E}_k),k=1,\dots,99$ with one such topology shown in Figure \ref{fig:star_cycle}, which satisfy $\mathrm{E}_k=\{(1,2),(1,3),\dots,(1,101),(l+1,l+2)\},l=1,\dots,k$. Note that $\mathrm{G}_k$ has $k$ $3$-node cliques, denoted by $\CC_l=\{1,l+1,l+2\},l=1,\dots,k$. Let all $100+k$ edges of $\mathrm{G}_k$ be a standard gossip averaging sequence in a fixed but arbitrarily chosen order, whose period-based state transition matrix is denoted by $\Fb_g$. By replacing $(l+1,l+2)$ with cliques $\CC_l,l=1,\dots,k,k=1,\dots,99$, we obtain a clique-gossip averaging sequence with its period-based state transition matrix denoted by $\Fb_c$. Evidently, the clique-gossiping and standard gossiping share the same period length. Then we plot $|\lambda_2(\Fb_g)|$ and $|\lambda_2(\Fb_c)|$ for values of $k=1,\dots,99$ in Figure \ref{fig:lambda2_star_cycle}. Since the period lengths for clique-gossiping and standard gossiping are equal, $|\lambda_2(\Fb_g)|$ and $|\lambda_2(\Fb_c)|$ embody their convergence speeds, respectively. We can see that the convergence speed of clique-gossiping is observably faster than standard gossiping. Moreover, the performance improvement becomes greater as the number of $3$-node cliques involved increases.

\begin{figure}[h]
\centering
\includegraphics[width=2.3in]{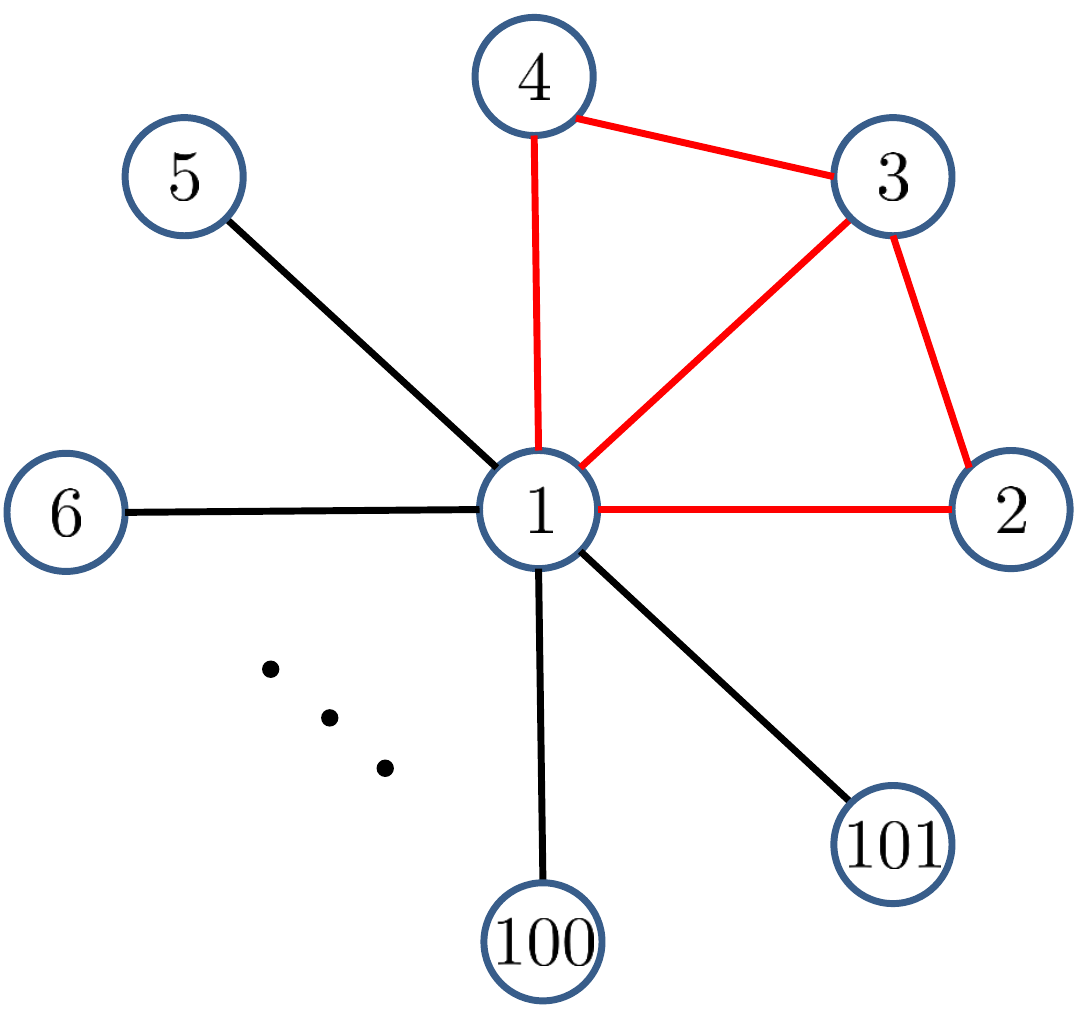}
\caption{A $101$-node graph $\mathrm{G}_k,k=2$.}
\label{fig:star_cycle}
\end{figure}

\begin{figure}[h]
\centering
\includegraphics[width=4in]{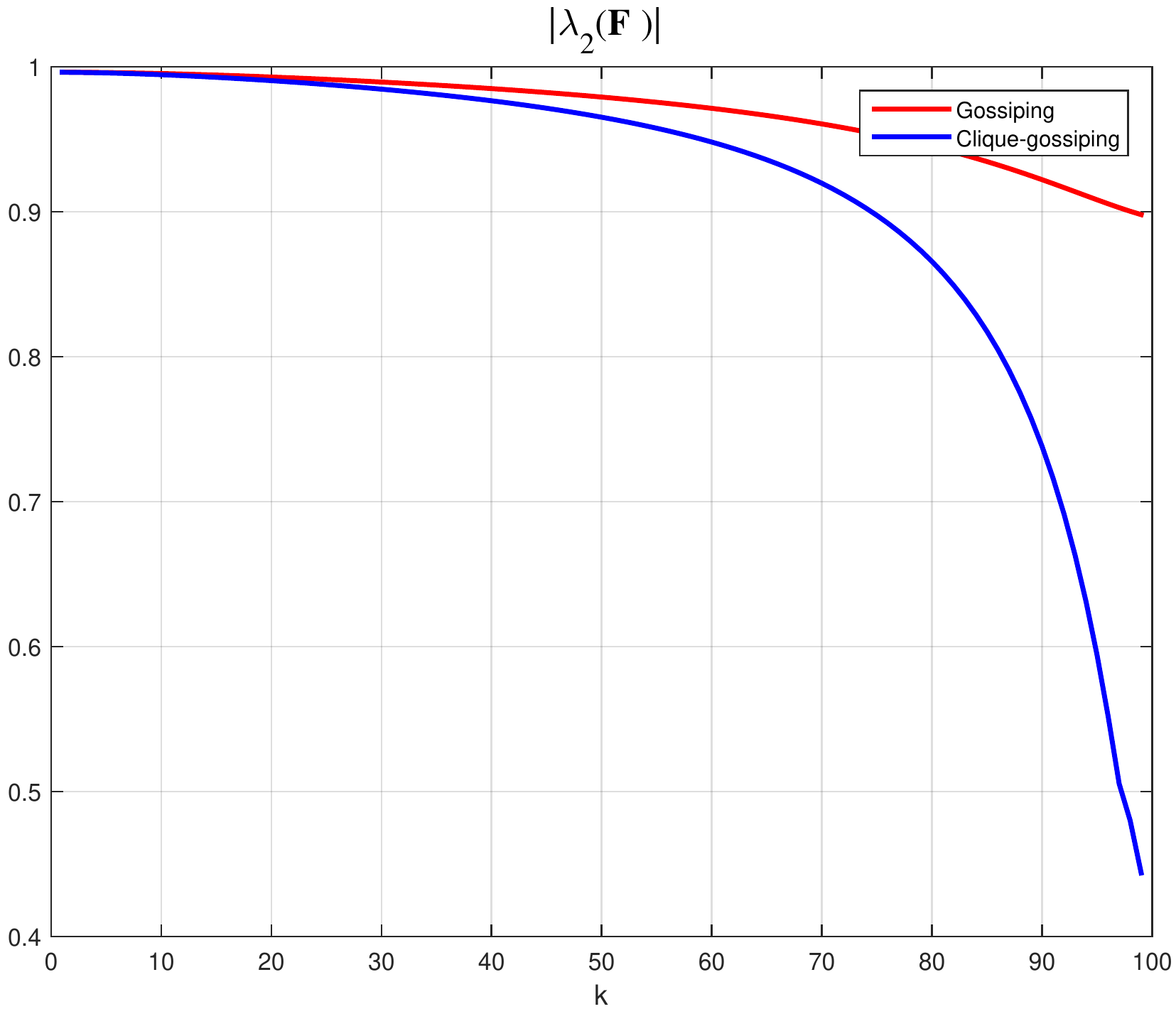}
\caption{The values of $|\lambda_2(\Fb_g)|$ and $|\lambda_2(\Fb_c)|$ for standard gossiping and clique-gossiping varying with $k=1,\dots,99$. Since they have the same period length, we can conclude that clique-gossiping has faster convergence speed than standard gossiping.}
\label{fig:lambda2_star_cycle}
\end{figure}

It is implied from Proposition \ref{prop:convergence_rate} that the convergence speed for periodic standard gossiping or clique-gossiping is determined by two factors: the period length and the second largest eigenvalue magnitude of the period-based state transition matrix. Clique-gossiping may provide faster convergence than standard gossiping by reducing the period length (as verified  in Example 5) or decreasing the the second largest eigenvalue magnitude (as verified  in Example 6).  An intriguing phenomenon observed from these two examples lies in that the performance improvement becomes more pronounced when we replace more edges in gossip sequence with cliques of size greater than two, and reduce the number of the nodes unable to be covered by cliques. Therefore, we conjecture that it is always encouraged to replace a pure gossip with a clique-gossip for the graphs containing cliques, in order to speed up distributed computation and make improvement in the algorithm performance. However, it is difficult to prove that clique-gossiping is more efficient than standard gossiping in a general case, because making comparison among the second largest eigenvalues  in magnitude of different period-based state transition matrices is a difficult problem.

\section{Clique-gossip Algorithms with Finite-time Convergence}\label{sec:finitetime}
In this section, we investigate the clique-gossip averaging algorithm introduced in Definition \ref{def:algorithm}. Formally the algorithm is written as
\begin{equation}\label{eqn:clique-gossiping}
\xb(t+1) = \Mb_{\sigma(t)} \xb(t)
\end{equation}
where $\xb(t)=(\xb_1(t)\dots \xb_n(t))^\top$ and $\Mb_{\sigma(t)}$ is induced by the matrices $\Ab_{ij}(\sigma(t))=1/|\CC_{\sigma(t)}|$. The asymptotic convergence of this algorithm has been clear from Lemma \ref{lem:convergence}.  Interestingly enough for standard gossip algorithms,  finite-time convergence is possible providing a definitive solution within a finite time steps \cite{Guodong2016ton}. Inspired by this we now study the finite-time convergence of clique-gossip algorithms. First we introduce the following definition.
\begin{definition}
 A clique-gossip averaging algorithm achieves finite-time convergence with respect to initial value $\xb(0)=\mathbf{c}\in \mathbb{R}^n$, if there exists a nonnegative integer $T$  (which may depend on $\mathbf{c}$) such that  $
 \xb(T)\in \mathrm{span}\{{\bf 1}\}.
  $
\end{definition}

Naturally, we say a clique-gossiping averaging algorithm achieves global finite-time convergence if finite-time convergence can be reached for any initial value in $\mathbb{R}^n$.
A feasible process of producing global finite-time convergence is provided in the following example.

\noindent{\bf Example 7.} Consider a node set $\mathrm{V}=\{1,2,\dots,12\}$ shown in Figure \ref{fig:finite_conv_nonregular}. Let $\CC_1=\{1,7\},\CC_2=\{2,8\},\CC_3=\{3,9\},\CC_4=\{4,10\},\CC_5=\{5,11\},\CC_6=\{6,12\},\CC_7=\{1,2,\dots,6\},\CC_8=\{7,8,\dots,12\}$. Suppose the node set of graph $\mathrm{G}=\bigcup\limits_{i=1}^8\mathrm{G}[\CC_i]$ is $\mathrm{V}$. Let $\mathsf{H}_{_{\mathrm{G}}}^\ast=\{\CC_1,\CC_2,\dots,\CC_8\}$ be a clique coverage of $\mathrm{G}$. It is evident that by performing averaging operations on first $\CC_1,\CC_2,\dots,\CC_6$ in an arbitrary order, then $\CC_7,\CC_8$ in an arbitrary order (or first $\CC_7,\CC_8$ in an arbitrary order, then $\CC_1,\CC_2,\dots,\CC_6$ in an arbitrary order), global finite-time convergence can be achieved over $\mathrm{G}$.

\begin{figure}[H]
\centering
\includegraphics[width=4.2in]{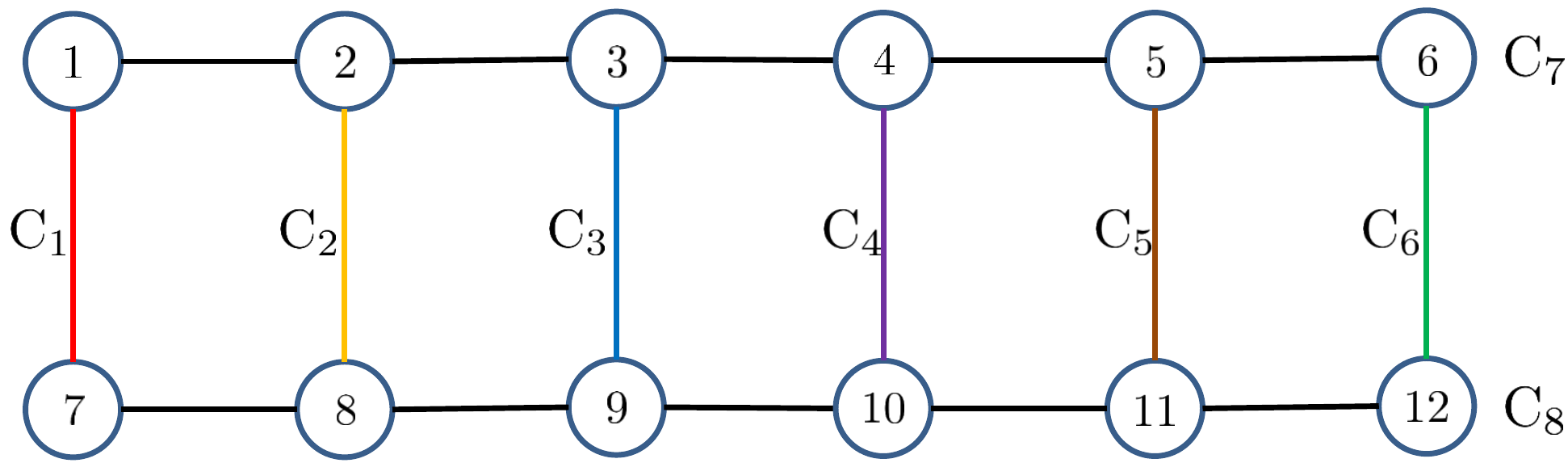}
\caption{A $12$-node complete graph $\mathrm{G}$ (only a subset of the edges are shown) with a clique coverage $\mathsf{H}_{_{\mathrm{G}}}^\ast=\{\CC_1,\CC_2,\dots,\CC_8\}$, where $\CC_1=\{1,7\},\CC_2=\{2,8\},\CC_3=\{3,9\},\CC_4=\{4,10\},\CC_5=\{5,11\},\CC_6=\{6,12\},\CC_7=\{1,2,\dots,6\},\CC_8=\{7,8,\dots,12\}$.}
\label{fig:finite_conv_nonregular}
\end{figure}

As can be seen in Example 7, the number of nodes $n=12=6\times 2$. As a result, global finite-time convergence can be achieved in $2+6=8$ steps by constructing two cliques of size $6$ and six cliques of size $2$. Inspired by Example 7, we present a sufficient condition for finite-time convergence in the following theorem.

\begin{theorem}\label{thm:pq}
Consider a node set $\mathrm{V}=\{1,2,\dots,n\}$ with $n=r_1r_2$ for two integers $r_1,r_2\ge 2$. Then there exists a graph $\mathrm{G}$ with its node set being $\mathrm{V}$ and a clique coverage $\mathsf{H}_{_{\mathrm{G}}}^\ast$ that consists of only cliques with sizes $r_1$ or $r_2$ leading to a globally finite-time convergent clique-gossip averaging algorithm. Furthermore, such finite-time convergence can be achieved in $r_1+r_2$ steps.
\end{theorem}
\begin{proof}
Define cliques $\CC_p=\{r_1(p-1)+1,r_1(p-1)+2,\dots,r_1p\},p=1,\dots,r_2$ and $\mathrm{Q}_q=\{q,r_1+q,\dots,r_1(r_2-1)+q\},q=1,\dots,r_1$. Let $\mathrm{G}=(\bigcup\limits_{i=1}^{r_2}\mathrm{G}[\CC_i]) \mcup (\bigcup\limits_{i=1}^{r_1}\mathrm{G}[\mathrm{Q}_i])$. Next we prove that along the $r_1+r_2$ long sequence of cliques $\CC_1,\dots,\CC_{r_2},\mathrm{Q}_1,\dots,\mathrm{Q}_{r_1}$, the algorithm yields a global finite-time convergence. Note that the cliques $\CC_p,p=1,\dots,r_2$ (or $\mathrm{Q}_q,q=1,\dots,r_1$) are mutually disjoint.
Suppose every node $i$ holds the initial state $\xb_i(0)$. After undertaking averaging operations over $\CC_p$s, we have the node $i\in\CC_p$'s state at time $t=r_2$
\begin{equation}\label{eq:pq1}
\xb_i(r_2)=\frac{1}{r_1}\sum\limits_{j\in\CC_p}\xb_j(0).
\end{equation}
Then we perform averaging operations over $\mathrm{Q}_q$s and one has for node $i\in\mathrm{Q}_q$
\begin{equation}\label{eq:pq2}
\xb_i(r_1+r_2)=\frac{1}{r_2}\sum\limits_{j\in\mathrm{Q}_q}\xb_j(r_2).
\end{equation}
It is worth noting that every node $j$, contained in the same $\mathrm{Q}_q$, belongs to a distinct $\CC_p$. Thus by (\ref{eq:pq1}) and (\ref{eq:pq2}), and the fact that $\CC_p$s are mutually disjoint
\begin{equation}\notag
\xb_i(r_1+r_2)=\frac{1}{r_1r_2}\sum\limits_{j=1}^n\xb(0),\ \forall i=1,\dots,n.
\end{equation}
This completes the proof.
\end{proof}

\begin{remark} Let us consider the case where  $n=r_1r_2\dots r_k$ with integers $r_1,r_2,\dots,r_k\ge 2$. By recursively applying Theorem \ref{thm:pq},  a   clique sequence with the cliques' sizes being $r_1,r_2,\dots,r_k$ can be constructed along which finite-time convergent averaging algorithm is defined with convergence achieved in    $$
\sum\limits_{i=1}^k \prod\limits_{j=1,j\neq i}^k r_j$$
steps. The intuition is that one can embed the $n$ nodes into a $k$-dimensional lattice  with the $j$'th dimension containing $r_j$ nodes. Then finite-time convergence can be built along each dimension. In particular, when $n=2^k$, the clique coverage $\mathsf{H}_{_{\mathrm{G}}}^\ast$ with all cliques being gossip edges can be found to produce finite-time convergence, as is known from \cite{Guodong2016ton}.

\end{remark}

Theorem \ref{thm:pq} provides the method of constructing a clique sequence for finite-time convergence, on condition that the total number of a graph's nodes is the product of two integers greater than one, which are exactly the size of the cliques to be constructed. In practical engineering problems, however, the number of the nodes contained in each selected clique is required to be unchanged, for the convenience of synchronization, noise computation, delay elimination, etc. In order to analyze the finite-time convergence in this background, we first provide the following definition.


\begin{definition}
A clique coverage  $\mathsf{H}_{_{\mathrm{G}}}^\ast$ for a graph $\mathrm{G}$ is  {\em $m$-regular} if every clique in $\mathsf{H}_{_{\mathrm{G}}}^\ast$ possesses exactly  $m$   nodes. The resulting clique-gossip averaging algorithm is called an $m$-regular clique-gossip averaging algorithm.
\end{definition}

It is obvious that not all connected graphs have an $m$-regular clique coverage if $m\geq 3$. For a complete graph with $n$ nodes, there always exists an $m$-regular clique coverage of the graph for any $m\leq n$. Now we are interested in the \emph{finite-time} convergence of  $m$-regular clique-gossip averaging algorithms. We present the following theorem.

\begin{theorem}\label{thm:finite-converge}
Let $\mathrm{V}=\{1,2,\dots,n\}$. There exists a graph $\mathrm{G}$ with its node set being $\mathrm{V}$ such that one can find an $m$-regular clique coverage $\mathsf{H}_{_{\mathrm{G}}}^\ast$ which can lead to a globally  finite-time convergent  clique-gossip averaging algorithm if and only if $n$ is divisible by $m$ with the same prime factors as $m$. More precisely, the following statements hold.
\begin{itemize}
\item[(i)] If $n$ is not a multiple of $m$,  or $n$ contains a different prime factor compared to $m$, then   no     $m$-regular clique-gossip averaging algorithm  converges  globally in finite time. In fact, in that case any given   $m$-regular clique-gossip averaging algorithm fails to converge in finite time for almost all initial values.

    \item[(ii)] Suppose there exist factorizations $m=p_1^{r_1}\cdots p_d^{r_d}$ and $n=p_1^{s_1}\cdots p_d^{s_d}$ with $p_1,\dots,p_d$ being prime numbers and $s_i\geq r_i>0$ for all $1\leq i\leq d$. Then there exists a globally convergent $m$-regular  clique-gossip averaging algorithm. Moreover, a fastest  $m$-regular  clique-gossip averaging algorithm converges in
        $$
       n \Big(\max_{1\leq i\leq d}\left\lceil \frac{s_i}{r_i}\right\rceil\Big)/m
        $$
        steps.
\end{itemize}
\end{theorem}

Theorem \ref{thm:finite-converge} is a generalization of the results on finite-time convergence with a standard gossip averaging algorithm~\cite{Guodong2016ton}, which corresponds to the special case of $m=2$. The sufficiency proof of the theorem is based on a constructive algorithm, where clearly only a small fraction of edges in the complete graph has been used. Therefore, the usefulness of this finite-time convergent result is not restricted only to the complete graph case. Finite-time convergence is also  possible if we  allow the $\Ab_{ij}(\sigma(t))$ to be genuinely time-dependent, e.g., \cite{Julien2015tac}, which will result in a consensus algorithm with time-varying state transitions. Now we provide an example to illustrate the finite-time convergence in Theorem \ref{thm:finite-converge}.

\noindent{\bf Example 8.} Consider the $18$-node complete graph $\mathrm{G}$ in Figure \ref{fig:finite_conv} (only a subset of the edges are shown). Let $\CC_1=\{1,2,3,4,5,6\},\CC_2=\{7,8,9,10,11,12\},\CC_3=\{13,14,15,16,17,18\},\CC_4=\{1,2,7,8,13,14\},\CC_5=\{3,4,9,10,15,16\},\CC_6=\{5,6,11,12,17,18\}$ and $\mathsf{H}_{_{\mathrm{G}}}^\ast=\{\CC_1,\CC_2,\dots,\CC_6\}$ be a $6$-regular clique coverage of $\mathrm{G}$. Also we plot its generalized line graph $\mathcal{L}(\mathsf{H}_{_{\mathrm{G}}}^\ast)$ in Figure \ref{fig:finite_conv_line}. Note that $\mathcal{L}(\mathsf{H}_{_{\mathrm{G}}}^\ast)$ is a complete bipartite graph. It can be seen that $\mathrm{G}$ with the clique coverage $\mathsf{H}_{_{\mathrm{G}}}^\ast$ satisfies the finite-time convergence condition in Theorem \ref{thm:finite-converge} and one can indentify that $n=18=2\times 3^2,m=2\times 3$. By undertaking the averaging operations on first $\CC_1,\CC_2,\CC_3$ in an arbitrary order, then $\CC_4,\CC_5,\CC_6$ in an arbitrary order (or first $\CC_4,\CC_5,\CC_6$ in an arbitrary order, then $\CC_1,\CC_2,\CC_3$ in an arbitrary order), the clique-gossip averaging algorithm yields finite-time convergence regardless of initial states.

\begin{figure}[H]
\centering
\includegraphics[width=4.2in]{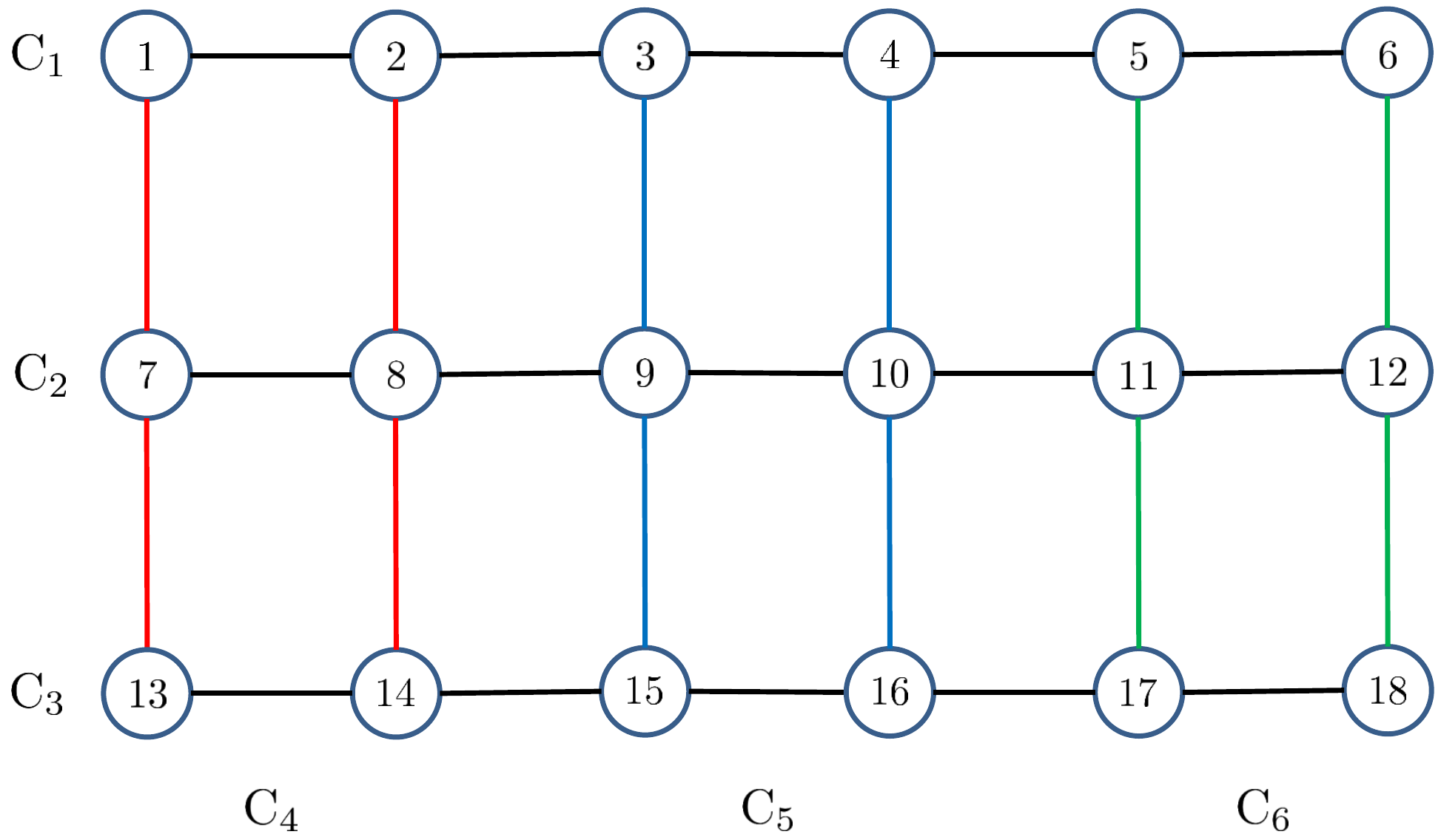}
\caption{A $18$-node complete graph $\mathrm{G}$(only a subset of the edges are shown) with a $6$-regular clique coverage $\mathsf{H}_{_{\mathrm{G}}}^\ast=\{\CC_1,\CC_2,\dots,\CC_6\}$, where $\CC_1=\{1,2,3,4,5,6\},\CC_2=\{7,8,9,10,11,12\},\CC_3=\{13,14,15,16,17,18\},\CC_4=\{1,2,7,8,13,14\},\CC_5=\{3,4,9,10,15,16\},\CC_6=\{5,6,11,12,17,18\}$.}
\label{fig:finite_conv}
\end{figure}

\begin{figure}[H]
\centering
\includegraphics[width=2.5in]{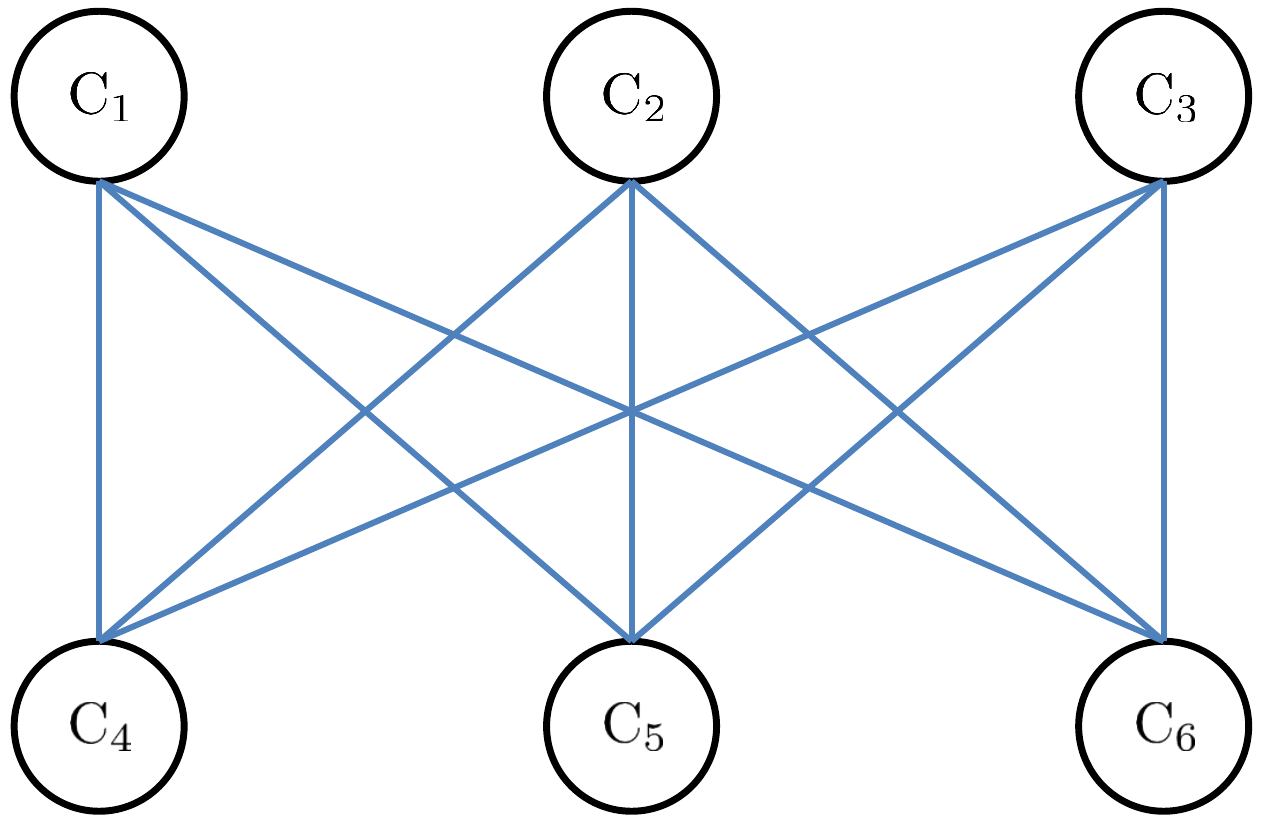}
\caption{The generalized line graph $\mathcal{L}(\mathsf{H}_{_{\mathrm{G}}}^\ast)$ for $\mathrm{G}$ in Figure \ref{fig:finite_conv}. It can be seen that $\mathcal{L}(\mathsf{H}_{_{\mathrm{G}}}^\ast)$ is a complete bipartite graph.}
\label{fig:finite_conv_line}
\end{figure}

\subsection{Proof of Sufficiency for Theorem \ref{thm:finite-converge}}
In this section, we prove that if $n$ is divisible by $m$ with the same prime factors as $m$, then there exists a graph $\mathrm{G}$ with its node set being $\mathrm{V}$ such that one can find an $m$-regular clique coverage $\mathsf{H}_{_{\mathrm{G}}}^\ast$ which can lead to a globally  finite-time convergent  clique-gossip averaging algorithm.

Let $m=p_1^{r_1}\cdots p_d^{r_d}$ and   $n=p_1^{s_1}\cdots p_d^{s_d}$ with $s_i\geq r_i>0$. We introduce $$
\bm{\delta}(n,m):=\Big(\max\limits_{1\leq i\leq d}\left\lceil \frac{s_i}{r_i}\right\rceil\Big).
$$
We denote by $(\mathrm{Q}_1\dots \mathrm{Q}_{l_1})$ a finite sequence of cliques of length $l_1$, where $\mathrm{Q}_i\in \mathsf{H}_{_{\mathrm{G}}}^\ast, 1\leq i \leq l_1$. Let $(\mathrm{Q}'_1\dots \mathrm{Q}'_{l_2})$ be another finite sequence of cliques of length $l_2$. We define the concatenation of  $(\mathrm{Q}_1\dots \mathrm{Q}_{l_1})$ and $(\mathrm{Q}'_1\dots \mathrm{Q}'_{l_2})$ as
\[
(\mathrm{Q}_1\dots \mathrm{Q}_{l_1}) \circ (\mathrm{Q}'_1\dots \mathrm{Q}'_{l_2}) = (\mathrm{Q}_1\dots \mathrm{Q}_{l_1}\mathrm{Q}'_1\dots \mathrm{Q}'_{l_2}),
\]
which is a finite sequence of length $l_1+l_2$. We now present a recursive  algorithm as a clique selection process over the  complete graph $\mathrm{G}=(\mathrm{V},\mathrm{E})$, the output of which is a finite sequence of cliques with $m$ nodes.

\begin{algorithm}[htb]
{$\mathbf{CliqueSelect}(\mathrm{V}, m)$}
\begin{algorithmic}[1]
\STATE  Let $n_1=\frac{n}{m}$.
\STATE If $n_1=1$, {\bf return} $\mathrm{V}$.
\STATE Otherwise, let $\mathrm{Q}_i=\big\{{m(i-1)+1},\dots, {mi}\big\}$ for $i=1,\dots,n_1$.
\STATE Let $m_1=\min\{b\in\mathbb{N}: m\mid n_1b\}$. Denote $n_2=\frac{m}{m_1}$.
\STATE Let $\mathrm{Q}^*_{ij}= \big\{{m(i-1)+m_1(j-1)+1},\dots, {m(i-1)+m_1j} \big\}$ for $j=1,\dots, n_2, i=1,\dots,n_1.$
\STATE Let $\mathrm{Q}^*_j =   \bigcup\limits_{i=1}^{n_1} \mathrm{Q}^*_{ij}, j=1,\dots,n_2.$
\RETURN $(\mathrm{Q}_1\dots \mathrm{Q}_{n_1})\circ \mathbf{CliqueSelect}(\mathrm{Q}^*_1, m)\circ\cdots\circ \mathbf{CliqueSelect}(\mathrm{Q}^*_{n_2},m)$;
\end{algorithmic}
\end{algorithm}

We first show this algorithm is well defined. Note that the following mathematical notations are all defined in the algorithm above. From the expressions $m=p_1^{r_1}\cdots p_d^{r_d}$ and $n=p_1^{s_1}\cdots p_d^{s_d}$ we know
$n_1 = p_1^{s_1-r_1}\cdots p_d^{s_d-r_d}$. By the definition of $m_1$, there holds  $m_1 =  p_1^{r'_1}\cdots p_d^{r'_d}$, where
\begin{align*}
&r'_1 = \max\{0, r_1-(s_1-r_1)\}\leq r_1,\\
&\dots \dots\\
&r'_d = \max\{0, r_d-(s_d-r_d)\}\leq r_d.
 \end{align*}
 Therefore, $m_1\mid m$, which implies that $n_2$ is a well defined  integer. We further know that each $\mathrm{Q}^*_{ij}$ contains $m_1$ nodes, and each $\mathrm{Q}^*_j$ contains $m_1n_1$ nodes. The definition of $m_1$ ensures $m\mid m_1n_1$ with $m_1n_1$ containing no distinct  prime factor compared to $m$. That is to say,
  $$\mathbf{CliqueSelect}(\mathrm{Q}^*_{1},m),\ \dots,\ \mathbf{CliqueSelect}(\mathrm{Q}^*_{n_2},m)
  $$ can be reasonably recursively invoked.

Next,  we prove by an induction argument  that the  clique sequence   produced by $
\mathbf{CliqueSelect}(\mathrm{V}, m)
$ is of length $ {\bm{\delta}(n,m)n}/{m}$, and the resulting  clique-gossip algorithm  converges in $ {\bm{\delta}(n,m)n}/{m}$ time steps. We complete the remainder of the proof in three steps.

\noindent {Step 1}. There holds $n=m$ if  $\bm{\delta}(n,m)=1$. The $\mathbf{CliqueSelect}(\mathrm{V}, m)$ returns one clique $\mathrm{V}$, and obviously the resulting clique-gossip algorithm converge in one step.  Now   we assume

\noindent{\em Induction Hypothesis}. For  $\bm{\delta}(n,m)\leq K-1$ with $K>1$,  $
\mathbf{CliqueSelect}(\mathrm{V}, m)
$ generates a sequence of $ {\bm{\delta}(n,m)n}/{m}$ cliques, along which the resulting  clique-gossip algorithm  converges in $ {\bm{\delta}(n,m)n}/{m}$ time steps.

\noindent {Step 2}.  Let $\bm{\delta}(n,m)=K>1$. Note that every clique selected by $\mathbf{CliqueSelect}(\mathrm{Q}^*_{j},m)$ contains $m_1n_1$ nodes. By the definition of $m_1$ and $n_1$, we can verify that $m_1n_1=p_1^{s'_1}\cdots p_d^{s'_d}$, where $$
s'_1 = \max\{r_1,s_1-r_1\},\dots,s'_d = \max\{r_d, s_d-r_d\}.
$$ This implies   $\bm{\delta}(n_1m_1,m)=K-1
 $, and by our induction hypothesis, each $\mathbf{CliqueSelect}(\mathrm{Q}^*_{j},m)$ produces  a sequence of $\frac{(K-1)n_1m_1}{m}$
cliques. Thus, the total length of the sequence $\mathbf{CliqueSelect}(\mathrm{V}, m)$ is
\[
n_1 + \frac{(K-1)n_1m_1}{m} n_2 = \frac{Kn}{m}.
\]
This establishes the number of cliques generated by the algorithm $\mathbf{CliqueSelect}(\mathrm{V}, m)$.

\noindent{Step 3}. We finally prove finite-time  convergence of the resulting clique-gossip algorithm along the clique sequence $\mathbf{CliqueSelect}(\mathrm{V}, m)$. Fix the initial value at all nodes. Then after the first $n_1$ steps, all nodes in $\mathrm{Q}_i=\big\{v_{m(i-1)+1},\dots, v_{mi}\big\}$ hold a common value $z_i$   for $i=1,\dots,n_1$.

Note that each $\mathrm{Q}_i$ is decomposed as $n_2$ disjoint subsets $\mathrm{Q}^*_{ij}, j=1,\dots, n_2$, where each $\mathrm{Q}^*_{ij}$ contains $m_1$ nodes. Therefore, at time $n_1$ and for $i=1,\dots,n_1$, there are $m_1$ nodes which hold value $z_i$ in $\mathrm{Q}^*_j  =\bigcup\limits_{i=1}^{n_1} \mathrm{Q}^*_{ij}$. Because the $\mathrm{Q}^*_j$ are mutually disjoint, the clique-gossip algorithm given by $\mathbf{CliqueSelect}(\mathrm{Q}^*_{j},m)$ does not influence the values of nodes outside  $ \mathrm{Q}^*_j$. Again by our induction hypothesis,  for any $j=1,\dots,n_2$, the clique-gossip algorithm given by $\mathbf{CliqueSelect}(\mathrm{Q}^*_{j},m)$  ensures that all nodes in $\mathrm{Q}^*_j$ hold  the same value $$
\frac{1}{n_1}\sum\limits_{i=1}^{n_1}z_i=\frac{1}{n}\sum_{j=1}^n \mathbf{x}_i(0).
 $$
 Therefore, all nodes in $\mathrm{V}$ will hold the same value as the average of the network initial values along the  clique-gossip algorithm generated  by $\mathbf{CliqueSelect}(\mathrm{V}, m)$ after $ {\bm{\delta}(n,m)n}/{m}$ time steps.
$\hfill\square$

\subsection{Proof of Necessity for Theorem \ref{thm:finite-converge}}
Now we prove that if there exists a graph $\mathrm{G}$ with its node set being $\mathrm{V}$ such that one can find an $m$-regular clique coverage $\mathsf{H}_{_{\mathrm{G}}}^\ast$ which can lead to a globally  finite-time convergent  clique-gossip averaging algorithm, then $n$ is divisible by $m$ with the same prime factors as $m$. We only need to find a particular initial value $\mathbf{c}\in \mathbb{R}^n$ that any deterministic clique-gossip algorithm will fail to converge in finite steps. We know that $m$ can be written uniquely as $ p_1^{r_1}\cdots p_k^{r_d}$, where $p_1<\cdots<p_d$ are prime numbers and $r_i>0,1\leq i\leq d$.

Given an arbitrary clique-gossip algorithm. We investigate two cases, respectively.
\begin{itemize}
  \item Let $n$ have a prime factor $p$ that $m$ does not have. Choose the initial value $\mathbf{c}=(1,0,\dots, 0)^T$. For any $t$, it is easy to see that $$
      \mathbf{x}_i(t)=\frac{\alpha_i(t)}{\beta_i(t)},
       $$where $\alpha_i(t)$ and $\beta_i(t)$ are coprime integers with $\beta_i(t)$ having no prime factor that $m$ does not have. That is to say, $\beta_i(t)$ does not contain the prime factor $p$, for any $t$. The limit of $\mathbf{x}_i(t)$ however must be $\sum_{i=1}^n\mathbf{x}_i(0)/{n}=1/n$. Because $n$ has $p$ as its prime factor,  such a value cannot be reached at any finite time steps.
  \item Let $n$  be represented by $p_1^{s_1}\cdots p_d^{s_d}$, where $s_i\geq 0, 1\leq i\leq d$ with some $1\leq a\leq d$ that $s_a<r_a$. Again we choose the initial value $\mathbf{c}=(1,0,\dots, 0)^T$. Similarly, for any $t$ we have $$
      \mathbf{x}_i(t)=\frac{\alpha_i(t)}{\beta_i(t)},
       $$
 where $\alpha_i(t)$ and $\beta_i(t)$ are coprime integers with $\beta_{i}(t)$ being some multiple of $p_a^{r_a}$.  Because $p_a^{r_a}$ cannot divide $n$ and the node state limit must be $1/n$, finite-time convergence is impossible.
\end{itemize}

\subsection{Proof of Almost Everywhere Impossibility}
Note that Theorem \ref{thm:finite-converge}.(i) asserts a stronger non-existence of claim in that any $m$-regular clique gossip algorithm fails to reach agreement in finite steps for almost all initial values.
Let $\mathscr{M}$ be a set consisting  of at most countable $n\times n$ real matrices. Define
\[
\mathbb{S}_{\mathscr{M}} = \big\{ \mathbf{c}\in \mathbb{R}^n: \exists t\geq 1, \mathbf{M}_1,\dots,\mathbf{M}_t\in \mathscr{M}, s.t. \ \mathbf{M}_t\cdots\mathbf{M}_1\mathbf{c} \in \mathrm{span}\{\bf 1\}\big\}.
\]
It is easy to verify that
\begin{equation*}
\mathbb{S}_{\mathscr{M}} = \bigcup\limits_{t=0}^{\infty}\bigcup\limits_{\mathbf{M}_0,\dots,\mathbf{M}_t\in \mathscr{M}}\mathbb{S}_{\mathbf{M}_0\dots\mathbf{M}_t},
\end{equation*}
where \[
\mathbb{S}_{\mathbf{M}_1\dots\mathbf{M}_t}=\big\{ \mathbf{c}\in \mathbb{R}^n:  \mathbf{M}_t\cdots\mathbf{M}_1\mathbf{c} \in \mathrm{span}\{\bf 1\}\big\}
\]
with $\mathbf{M}_1,\dots,\mathbf{M}_t\in \mathscr{M}$.

Note that each $\mathbb{S}_{\mathbf{M}_1\dots\mathbf{M}_t}$ is a linear subspace of $\mathbb{R}^n$, with a dimension no larger than $n$. If all $\mathbb{S}_{\mathbf{M}_1\dots\mathbf{M}_t}$ are lower-dimensional  subspaces of $\mathbb{R}^n$,  $\mathbb{S}_{\mathbf{M}_1\dots\mathbf{M}_t}$ has zero measure for any $\mathbf{M}_0,\dots,\mathbf{M}_t\in \mathscr{M}$. This in turn tells us that $\mathbb{S}_{\mathscr{M}}$ is a zero-measure set for $\mathscr{M}$ is a union of   countably many  zero-measure sets. On the other hand,  if  there exists $\mathbf{M}_0,\dots,\mathbf{M}_t\in \mathscr{M}$ such that $\mathbb{S}_{\mathbf{M}_1\dots\mathbf{M}_t}$ is  $n$ dimension, we have $\mathbb{S}_{\mathscr{M}}=\mathbb{R}^n$.  Therefore,  either $\mathbb{S}_{\mathscr{M}}=\mathbb{R}^n$ or $\mathbb{S}_{\mathscr{M}}$ is a zero-measure set in $\mathbb{R}^n$. The desired almost everywhere impossibility conclusion holds immediately since we already proved non-existence of globally finite-time  convergent $m$-regular clique gossiping.

\subsection{Proof of Complexity}

Recall that $m=p_1^{r_1}\cdots p_d^{r_d}$ and   $n=p_1^{s_1}\cdots p_d^{s_d}$ with $s_i\geq r_i>0$. We have provided an algorithm that converges in ${\bm{\delta}(n,m)n}/{m}$ steps. Now we prove that it is indeed  the fastest algorithm. Consider any $m$-regular clique gossip algorithm that converges globally in finite time. Then  there must exist $T\geq 0$ such that
\[
\mathbf{M}_{\sigma(T)}\cdots\mathbf{M}_{\sigma(1)} = \frac{1}{n} {\bf 1}{\bf 1^T}.
\]
Introduce $\mathbf{N}(t) = \mathbf{M}_{\sigma(t)}\cdots\mathbf{M}_{\sigma(1)}$ and $
h_i(t) = \Big|\big\{s: [\mathbf{M}_{\sigma(s)}]_{ii} = \frac{1}{m}, 1\leq s\leq t\big\}\Big|$.
Note that  $h_i(t)$ represents  the number of times at which $i$  is in the selected cliques for the first $t$ steps.

Denote
$\tau=\arg\max_{1\leq i\leq d} \lceil \frac{s_i}{r_i}\rceil$. Associated with the prime number  $p_{\tau}$, we define a function  $\mathpzc{P}_\tau(\cdot)$ over all natural numbers by
\[
\mathpzc{P}_\tau(x) = \max\{k: x=yp_{\tau}^k,y\in \mathbb{N}, k\in \mathbb{Z}\}.
\]
In other words, $\mathpzc{P}_\tau(x)$ is the number of powers of the prime number $p_{\tau}$ in the arithmetic decomposition of $x$.

We can verify recursively that $$
[\mathbf{N}(t)]_{ii}=\frac{\gamma_{ii}(t)}{\delta_{ii}(t)},
 $$where $\gamma_{ii}(t)$ and $\delta_{ii}(t)$ are coprime numbers with $\mathpzc{P}_\tau(\delta_{ii}(t))\leq h_i(t)r_{\tau}$. Based the facts that $[\mathbf{N}(T)] = {1}/{n}$ and $\mathpzc{P}_\tau(n)> \big(\bm{\delta}(n,m)-1 \big)r_{\tau}$, we obtain
\[
\big(\bm{\delta}(n,m)-1\big)r_{\tau}< \mathpzc{P}_\tau(n)= \mathpzc{P}_\tau(\delta_{ii}(T)) \leq h_i(T)r_{\tau}.
\]
This implies  $h_i(T)\geq \bm{\delta}(n,m)$. On the other hand, there must hold
\[
\bm{\delta}(n,m)n\leq \sum\limits_{i=1}^nh_i(T) = Tm.
\]
We can now conclude  $T\geq {\bm{\delta}(n,m)n}/{m}$, and this is the fundamental  lower bound  that any $m$-regular clique gossip algorithm can reach  in terms of convergence time. We have now proved the complexity claim in Theorem \ref{thm:finite-converge}.(ii).

\section{Conclusions}\label{sec:conclusions}
We have presented  a framework for clique gossip protocols where node interactions utilize cliques as  complete  subnetworks in gossip processes.     Clique-gossip protocols  and  clique-gossip averaging algorithms  have been defined as generalizations of standard gossip protocols and gossip averaging algorithms, respectively. A fundamental eigenvalue invariance principle  for   periodic clique-gossip protocols was established, and the possibilities of realizing finite-time convergent clique-gossip averaging were thoroughly investigated. Numerical examples also revealed the acceleration effects of clique-gossiping compared to standard gossiping.  Interesting future directions include concrete theoretical validations of how much improvement can be gained via clique-gossiping in terms of efficiency, and self-organized or engineering mechanism that produces local cliques across a network.

\end{document}